\documentclass[10pt, a4paper,reqno]{amsart}

\usepackage{color} \definecolor{bleu_sombre}{rgb}{0,0,0.6}  \definecolor{rouge_sombre}{rgb}{0.8,0,0}\definecolor{vert_sombre}{rgb}{0,0.6,0}
\usepackage[plainpages=false,colorlinks,linkcolor=bleu_sombre,citecolor=rouge_sombre,urlcolor=vert_sombre,breaklinks]{hyperref}

\usepackage[T1]{fontenc}	
\usepackage[latin1]{inputenc}	
\usepackage[english]{babel}

\usepackage{amsmath,amssymb,amsthm,amsfonts}
\usepackage{url,color}
\usepackage{enumerate,stmaryrd,dsfont}

\theoremstyle{plain}
\newtheorem{theorem}{{Theorem}}[section] 
\newtheorem*{theorem*}{{Theorem}}
\newtheorem{proposition}[theorem]{Proposition}
\newtheorem*{proposition*}{Proposition}
\newtheorem{corollary}[theorem]{Corollary}
\newtheorem*{corollary*}{Corollary}
\newtheorem{lemma}[theorem]{Lemma}
\newtheorem*{lemma*}{Lemma}

\theoremstyle{definition}
\newtheorem{definition}[theorem]{Definition}
\newtheorem*{definition*}{Definition}

\theoremstyle{remark}
\newtheorem*{remarque*}{Remarque}

\newtheorem{remark}[theorem]{Remark}

\newtheorem*{example*}{Example}

\newtheorem*{examples*}{Examples}

\newcommand{\commm}[1]{{}}

\makeatletter

\@addtoreset{equation}{section}  
\makeatother

\renewcommand{\leq}{\leqslant}	\renewcommand{\geq}{\geqslant}
\renewcommand{\bar}[1]{\overline{#1}}
\renewcommand\over[2]{{\,\buildrel #1\over#2\,}}

\newcommand{\inv}{^{-1}}

\newcommand{\abs}[1]{\left\vert #1\right\vert}        
\newcommand{\nr}[1]{\left\Vert #1\right\Vert}         
\newcommand{\innp}[2]{\left< #1 , #2 \right>}         

\newcommand{\Dom}{\Dc}			
\newcommand{\Op}{{\mathop{\rm{Op}}}_h}
\newcommand{\Opw}{{\mathop{\rm{Op}}}_h^w}

\newcommand{\pppg}[1] {\left< #1 \right>} 	
\newcommand{\symb} {\Sc}		

\newcommand{\bigo}[2]{\mathop{O}\limits_{#1 \to #2}}

\newcommand{\singl}[1]{\left\{ #1 \right\}}		
\newcommand{\Ii}[2]{\llbracket #1,#2 \rrbracket}	
\newcommand{\R}{\mathbb{R}}		\newcommand{\C}{\mathbb{C}}
\newcommand{\N}{\mathbb{N}}

\newcommand{\1}[1]{\mathds 1 _{#1}}

\newcommand{\st}{\,:\,}					

\newcommand{\restr}[2]{\left.#1\right|_{#2}}         
\renewcommand{\Re}{\mathop{\rm{Re}}\nolimits}        
\renewcommand{\Im}{\mathop{\rm{Im}}\nolimits}        
	
\DeclareMathOperator{\Id}{Id}                        
 
\DeclareMathOperator{\supp}{supp}                    

\renewcommand{\a}{\alpha}\renewcommand{\b}{\beta}\newcommand{\g}{\gamma}\renewcommand{\d}{\delta}\newcommand{\D}{\Delta}\newcommand{\e}{\varepsilon}\newcommand{\z}{\zeta} \newcommand{\y}{\eta}\renewcommand{\th}{\theta}\newcommand{\Th}{\Theta}\renewcommand{\k}{\kappa}\renewcommand{\l}{\lambda}\newcommand{\m}{\mu}\newcommand{\n}{\nu}\newcommand{\x}{\xi}\newcommand{\s}{\sigma}\renewcommand{\t}{\tau}\newcommand{\f}{\varphi}\newcommand{\vf}{\phi}\newcommand{\h}{\chi}\newcommand{\p}{\psi}\renewcommand{\O}{\Omega}

\newcommand{\Bc}{{\mathcal B}}\newcommand{\Dc}{{\mathcal D}}\newcommand{\Ec}{{\mathcal E}}\newcommand{\Hc}{{\mathcal H}}\newcommand{\Kc}{{\mathcal K}}\newcommand{\Lc}{{\mathcal L}}\newcommand{\Nc}{{\mathcal N}}\newcommand{\Rc}{{\mathcal R}}\newcommand{\Sc}{{\mathcal S}}\newcommand{\Tc}{{\mathcal T}}\newcommand{\Uc}{{\mathcal U}}

\newcommand{\ad}{{\rm{ad}}}

\newcommand{\divg}{\mathop{\rm{div}}\nolimits}

\newcounter{stepproof}
\newcommand{\stepp}{\stepcounter{stepproof} \noindent {\bf $\bullet$}\quad }

\begin{document}

\newcommand{\Da}{\pppg D^{\a}}
\newcommand{\bb}{\tilde \a}
\newcommand{\bbb}{{\gamma}}
\newcommand{\kk}{\kappa}

\newcommand{\ybar}{\bar \y}
\newcommand{\yy}{\y_1}
\newcommand{\yyy}{\y_2}
\newcommand{\yoo}{\tilde \y}

\newcommand{\Pg}{P}
\newcommand{\Pii}{P_{\yy}}
\newcommand{\Piio}{P_{\y_0}}
\newcommand{\tPii}{P_{\yy,z}}
\newcommand{\Pcc}{P_{\yy,c}}
\newcommand{\Pcco}{P_{\y_0,c}}
\newcommand{\Hg}{H}
\newcommand{\Rg}{R}
\newcommand{\Hii}{H_{\bar \y}}
\newcommand{\tHii}{H_{\ybar,z}}
\newcommand{\Hcc}{H_{\ybar,c}}
\newcommand{\Ba}{B_\a}
\newcommand{\Bai}{B_{\yyy}^\a}
\newcommand{\tBai}{B_{\yyy,z}^\a}
\newcommand{\Bac}{B_{\yyy,c}^\a}
\newcommand{\Ri}{R_{\bar \y}}
\newcommand{\Riz}{R_{\ybar,z}}
\newcommand{\tRi}{\tilde R_{\yyy}}
\newcommand{\Ro}{R_0}
\newcommand{\Thz}{\Th_z}
\newcommand{\Kz}{K_{\y_0}}
\newcommand{\Kzyy}{K_{\yy}}
\newcommand{\Sp}{S_{\p,\ybar}}
\newcommand{\Spz}{S_{\p,\ybar}(z)}

\newcommand{\Hh}{H_h}
\newcommand{\Ph}{P_h}
\newcommand{\tBa}{B_h^\a}

\newcommand{\BB}{\Bc_\p}

\author{Moez Khenissi and Julien Royer}

\title{Local energy decay and smoothing effect for the damped Schr\"odinger equation}

\begin{abstract}
We prove the local energy decay and the smoothing effect for the damped Schrödinger equation on $\R^d$. The self-adjoint part is a Laplacian associated to a long-range perturbation of the flat metric. The proofs are based on uniform resolvent estimates obtained by the dissipative Mourre method. All the results depend on the strength of the dissipation which we consider.
\end{abstract}

\maketitle

\section{Introduction}

Let $d \geq 3$. Our purpose in this paper is to study on $\R^d$ the local energy decay and the Kato smoothing effect for the damped Schr\"odinger equation 
\begin{equation} \label{eq-damp-schrodinger}
\begin{cases}
-i \partial_t u + \Pg u -i a(x) \Da a(x) u = 0,\\
u(0) = u_0.
\end{cases}
\end{equation}
The operator $\Pg$ is a Laplacian in divergence form (or a Laplace-Beltrami operator) associated to a metric which is a long-range perturbation of the usual flat metric (see \eqref{hyp-long-range} below). For the dissipative part we have denoted by $\pppg \cdot$ the function $(1+\abs \cdot ^2)^{\frac 12}$ and by $D$ the square root of the free Laplacian, so that $\Da$ stands for $(1-\D)^{\frac \a 2}$. The parameter $\a$ belongs to $[0,2[$. The non-negative valued function $a$ will be assumed to be of short range (see \eqref{hyp-a-short-range}), so that in terms of spacial decay, we have an absorption index $a(x)^2$ which decays at least like $\pppg x^{-2-2 \rho}$ for some $\rho > 0$.\\

It is known that the free Schrödinger equation (\eqref{eq-damp-schrodinger} with $\Pg = -\D$ and $a = 0$) preserves the $L^2$-norm but satisfies the local energy decay: if $u_0$ is supported in the ball $B(R) = \singl{\abs x \leq R}$ of $\R^d$ for some $R > 0$ we have
\[
\nr{e^{it\D}u_0}_{L^2(B(R))} \leq C_R \pppg t^{-\frac d 2} \nr{u_0}_{L^2(\R^d)}.
\]
This means that the ``mass'' of the solution escapes at infinity. On the other hand, the Schrödinger equation has a regularizing effect. The solution $e^{it\D} u_0$ belongs to $C^\infty$ for $t \neq 0$ and 
\[
\int_\R \nr{(1-\D)^{\frac 14} e^{it\D}u_0}^2_{L^2(B(R))} \, dt \leq C_R \nr{u_0}^2_{L^2(\R^d)}.
\]

There are many papers dealing with these properties for more and more general Schr\"odinger equations. 
Concerning the local energy decay for a self-adjoint Schr\"odinger equation, we only refer to \cite{rauch78} for the Schr\"odinger operator with an exponentially decaying potential, to \cite{tsutsumi84} for the free Schr\"odinger equation on an exterior demain, and to \cite{bouclet11,bonyh12} for a laplacian associated to a long-range perturbation of the flat metric. For all these papers, the local energy decays like $t^{-\frac d 2}$ or like $t^{-\frac d 2 + \e}$ under a non-trapping assumption. There is also a huge litterature for the closely related problem of the local energy decay for the wave equation (see \cite{laxmp63,ralston69,morawetzrs77,burq98,tataru13,guillarmouhs13} and references therein).

Concerning the smoothing effect we mention \cite{constantins88,sjolin87} for the laplacian on $\R^d$, \cite{benartzik92} for the Schr\"odinger operator with a potential and \cite{burqgt04} for the problem on an exterior domain. We also refer to \cite{doi96,doi00,burq04} for the necessity of the non-trapping condition.

In the dissipative context, the local energy decay for the damped Schr\"odinger equation in an exterior domain has been proved in \cite{alouik07}. In this context, the non-trapping condition can be replaced by the geometric control condition (see \cite{raucht74,bardoslr92}): there can be bounded classical trajectories but they have to go through the damping region. Then the local energy decays like $t^{-\frac d 2}$, as in the self-adjoint case under the non-trapping condition. A similar result has been obtained in \cite{alouik10} for the free Schr\"odinger equation on an exterior domain with dissipation at the boundary, and in \cite{art-diss-schrodinger-guide} for the similar problem on a wave guide. In the latter case the global energy decays exponentially and we have a smoothing effect in the unbounded directions. We also mention \cite{bortotc}, where an exponential decay for the global energy is proved for the solution of the Schr\"odinger equation with a dissipation 
effective on a neighborhood of the infinity.

The dissipation by a potential ($\a = 0$ in our setting) is not strong enough to recover under the damping condition the same smoothing effect as under the non-trapping condition. However it is known that this is the case for the so-called regularized Schr\"odinger equation ($\a = 1$). See \cite{aloui08,aloui08b} for the problem on a compact manifold and \cite{alouikr} for the problem on an exterior domain. As in the self-adjoint case (see \cite{burq04}), we can recover a $H^{\frac 12 - \e}$ smoothing effect if only a few classical trajectories fail to satisfy the assumption (see \cite{alouikv13}).

In these works, the problem is a compact perturbation of the free Schr\"odinger equation.  Our purpose in this paper is to prove the local energy decay and the Kato smoothing effect for an asymptotically vanishing perturbation. In a similar context, the local energy decay has been studied for the dissipative wave equation in \cite{boucletr14}.\\

We now describe more precisely the setting of our paper. We consider on $\R^d$ a metric $G(x)$ which is a long range perturbation of the identity: for some $\rho > 0$ there exist constants $C_\b$ for $\b \in \N^d$ such that
\begin{equation} \label{hyp-long-range}
\abs{\partial^\b \big( G(x) - I_d \big)} \leq C_\b \pppg x ^{-\rho - \abs \b}.
\end{equation}
Let $w \in C^\infty(\R^d)$ be such that  
\begin{equation} \label{hyp-w}
C\inv \leq w(x) \leq C
\end{equation}
for some $C \geq 1$. We consider $b_1,\dots,b_d \in C_0^\infty(\R^d)$ such that the operator
\begin{equation} \label{def-Pg}
\Pg = - \divg G(x) \nabla + W, \quad \text{where } W = \sum_{j=1}^d b_j(x) D_j,
\end{equation}
is self-adjoint and non-negative on $L^2_w := L^2(\R^d, w(x) \, dx)$ (with domain $H^2_w = H^2(\R^d,w(x) \, dx)$). Here and everywhere below, $D_j$ stands for $-i\partial_{x_j}$. For a Laplacian in divergence form we only have to set 
\begin{equation} \label{setting-divergence-form}
w = 1 , \quad b_1 = \dots = b_n = 0.
\end{equation}
We now turn to the Laplace-Beltrami operator associated to a metric $g$. It is defined as
\[
- \D_g = - \sum_{j,k=1}^d \abs{g(x)}^{-\frac 12} \frac \partial {\partial x_j} \abs{g(x)}^{\frac 12} G_{j,k}(x) \frac {\partial} {\partial x_k},
\]
with $\abs{g(x)} = \det (g_{j,k}(x))$ and $(G_{j,k}) = (g_{j,k})\inv$. If $g$ is a long-range perturbation of the flat metric, then so is $G = g\inv$. We recall from \cite{bouclet11} that we can assume without loss of generality that $\abs {g(x)} = 1$ outside a compact set of $\R^d$. Thus there exist $b_1,\dots b_d \in C_0^\infty(\R^d)$ such that $-\D_g$ is as in \eqref{def-Pg} with 
\begin{equation} \label{setting-laplace-beltrami}
w = \abs{g(x)}^{\frac 12} \quad \text{and} \quad G = g\inv.
\end{equation}

Concerning the dissipative term, $a$ is a smooth and non-negative valued function on $\R^d$. As already mentioned, it is of short range:
\begin{equation}\label{hyp-a-short-range}
\abs{\partial^\b a(x)} \leq C_\b \pppg x^{-1 - \rho -\abs \b}.
\end{equation}
We will use the following notation:
\begin{equation} \label{def-Ba-H}
\Ba =  a(x) \Da a(x) \qquad \text{and} \qquad \Hg = \Pg - i \Ba.
\end{equation}
We also set 
\begin{equation} \label{def-talpha-kappa}
\bb = \min(1,\a) \quad \text{and} \quad \kk = \begin{cases} \frac d 2 & \text{if $d$ is even},\\ \frac {d+1}2 & \text{if $d$ is odd}. \end{cases}
\end{equation}

We will see that $\Hg$ is a maximal dissipative operator on $L^2_w$. In particular, for $u_0 \in \Dom(\Hg) = H^2_w$ the problem \eqref{eq-damp-schrodinger} has a unique solution $t \mapsto e^{-it\Hg}u_0$. The main purpose of this paper is to prove that this solution satisfies the local energy decay and the Kato smoothing effect as stated in the following two theorems:

\begin{theorem}[Local energy decay] \label{th-loc-decay}
Let $\e > 0$. Let $\d > \k + \frac 12$, $N \in \N$ and $\s \in [0,2]$. Assume that 
\begin{enumerate}[(i)]
\item there are no bounded geodesics (see the non-trapping condition \eqref{hyp-non-trapping} below)
\item or the bounded geodesics go through the damping region (see \eqref{hyp-damping}), $N \bb + \s \geq 2$ and $\d > N - \frac 12$.
\end{enumerate}
Then there exists $C \geq 0$ such that for $u_0 \in H^{\s,\d}$ and $t \geq 0$ we have 
\[
\nr{e^{-it \Hg} u_0}_{L^{2,-\d}} \leq C t^{-\frac d 2 + \e} \nr{u_0}_{H^{\s,\d}}.
\]
\end{theorem}

In this statement $L^{2,-\d}$ denotes the weighted space $L^2\big(\pppg x^{-2\d} \,dx \big)$, while $\nr{u_0}_{H^{\s,\d}}$ is the $L^2$-norm of $\pppg x^\d \pppg D^\s u_0$.

We remark that we have to take $\s = 2$ in the second case if $\a = 0$. This means that we have a loss of two derivatives. If $\a > 0$ we can take $\s = 0$ (no loss of derivative) as long as we choose $\d$ large enough (if $\a \geq 1$ then we can take $N = 2$ and in this case the condition $\d > N - \frac 12$ is weaker than $\d > \k + \frac 12$). Under the non-trapping condition we can always take $\s = 0$.

\begin{theorem}[Kato smoothing effect] \label{th-smoothing-effect}
There exists $C \geq 0$ such that for all $u_0 \in L^2$ we have 
\[
\int_0^{+\infty} \nr{\pppg x^{-1} \pppg D^{\bb / 2} e^{-it \Hg} u_0}^2_{L^2} \, dt \leq C \nr {u_0}^2_{L^2}.
\]
\end{theorem}

Notice that according to \eqref{hyp-w} we have $L^{2,\d} = L^{2,\d}_w$ and $H^{\s,\d}_w = H^{\s,\d}$ with equivalent norms, so we can state these estimates in $L^2$ and $H^{\s,\d}$ even when $w \not \equiv 1$.\\

The proofs of Theorems \ref{th-loc-decay} and \ref{th-smoothing-effect} are based on uniform resolvent estimates. According to Proposition \ref{prop-max-diss-acc} below, the operator $\Hg$ is maximal dissipative, so for all $z$ in 
\[
\C_+ := \singl{z \in \C \st \Im z > 0}
\]
we can consider the resolvent
\begin{equation} \label{def-res}
\Rg(z) = \big(\Hg - z\big) \inv \in \Lc(L^2_w) = \Lc(L^2). 
\end{equation}
Here we have denoted by $\Lc(L^2)$ the space of bounded operators on $L^2$. After a Fourier transform, the solution $u$ of \eqref{eq-damp-schrodinger} can be written as the integral over frequencies $\Re(z)$ of this resolvent when $\Im(z)$ goes to 0 (see Section \ref{sec-loc-decay}). Thus the problem will be reduced to proving uniform estimates for $\Rg(z)$ and its derivatives for $\Im (z)$ small, and then to control the dependance of these estimates with respect to $\Re(z)$. Since the self-adjoint part $\Pg$ of $\Hg$ is a non-negative operator, the estimates for $\Re(z) < 0$ are easy: for $n \in \N$ and $z \in \C_+$ with $\Re(z) \leq -c_0 <  0$ we have
\begin{equation} \label{estim-Rg-trivial}
\nr {\Rg^{n+1}(z)}_{\Lc(L^2)} \leq \frac {C_w} {\abs{\Re(z)}^{n+1}}.
\end{equation}
Thus we will focus on $z \in \C_+$ with $\Re(z) \geq -c_0$ where $0 < c_0 \ll 1$. As usual, the difficulties will arise for low frequencies ($\Re(z)$ close to 0) and high frequencies ($\Re(z) \gg 1$). We first state the uniform resolvent estimates for intermediate frequencies:

\begin{theorem}[Intermediate frequency estimates] \label{th-inter-freq}
Let $K$ be a compact subset of $\C \setminus \singl 0$. Let $n \in \N$ and $\d > n + \frac 12$. Then there exists $C \geq 0$ such that for all $z \in K \cap \C_+$ we have 
\[
\nr{\pppg x^{-\d} \Rg^{n+1} (z)\pppg x ^{-\d}}_{\Lc(L^2)} \leq C.
\]
\end{theorem}

We remark that compared to the resolvent for the dissipative wave equation (see \cite{boucletr14}), the derivatives of the resolvent correspond to its powers:
\[
\Rg^{(n)}(z) = n! \, \Rg^{n+1} (z).
\]
This will significantly simplify the discussion.\\

It is known that even for the free Laplacian the estimates of Theorem \ref{th-inter-freq} fail to hold uniformly when $z$ goes to 0 if $n$ is too large. This explains the restriction in the rate of decay in Theorem \ref{th-loc-decay}. For low frequencies we prove the following result:

\begin{theorem}[Low frequency estimates]  \label{th-low-freq}
Let $\e > 0$. Let $n \in \N$ and let $\d$ be such that 
\[
\d >
\begin{cases}
n + \frac 12 & \text{if } 2n+1 \geq d,\\
n+1 & \text{if } 2n+1<d.
\end{cases}
\]
Then there exist $C\geq 0$ and a neighborhood $\Uc$ of $0$ in $\C$ such that for all $z \in \Uc \cap \C_+$ we have 
\[
\nr{\pppg x^{-\d} \Rg^{n+1}(z) \pppg x ^{-\d}}_{\Lc(L^2)} \leq C \left( 1 + \abs z ^{\frac d 2 - \e - 1 - n} \right).
\]
\end{theorem}

In the self-adjoint case we can improve the estimate for a single resolvent. More precisely we can replace the weight $\pppg x^{-\d}$ for $\d > 1$ by $\pppg x \inv$. See \cite{boucletr}. This is particularly interesting for Theorem \ref{th-smoothing-effect} which does not require estimates for the derivatives of the resolvent. This sharp resolvent estimates is also valid in our dissipative context:

\begin{theorem}[Sharp low frequency estimate]  \label{th-low-freq-sharp}
There exist $C\geq 0$ and a neighborhood $\Uc$ of $0$ in $\C$ such that for all $z \in \Uc \cap \C_+$ we have 
\[
\nr{\pppg x \inv \Rg(z) \pppg x \inv}_{\Lc(L^2)} \leq C.
\]
\end{theorem}

The high frequency properties of the problem are closely related to the corresponding classical problem. Here the classical flow is the geodesic flow on $\R^{2d} \simeq T^* \R^d$ for the metric $G(x)\inv$ (that is the geodesic flow of the metric $g$ when $\Pg = -\D_g$). It is the Hamiltonian flow corresponding to the symbol 
\[
p(x,\x) = \innp{G(x) \x}{\x}.
\]
We denote by $\vf^t = (X(t),\Xi(t))$ this flow. Let 
\[
\O_b = \singl{w \in p\inv(\singl 1) \st \sup _{t\in\R} \abs {X(t,w)} < +\infty}.
\]
We say that the classical flow is non-trapping if there is no bounded geodesic:
\begin{equation} \label{hyp-non-trapping}
\O_b = \emptyset.
\end{equation}
We say that the damping condition on bounded geodesics (or Geometric Control Condition) is satisfied if every bounded geodesic goes through the damping region $\singl{a(x) > 0}$:
\begin{equation} \label{hyp-damping}
\forall w \in \O_b, \exists T \in \R, \quad a \big(X(T,w)\big) > 0.
\end{equation}

\begin{theorem}[High frequency estimates]  \label{th-high-freq}
 Let $n \in \N$ and $\d > n + \frac 12$.

\begin{enumerate}[(i)]
\item Assume that the non-trapping assumption \eqref{hyp-non-trapping} holds. Then there exists $C \geq 0$ such that for $z \in \C_+$ with ${\Re(z)} \geq C$ we have 
\[
\nr{\pppg x ^{-\d} \Rg^{n+1}(z) \pppg x ^{-\d}}_{\Lc(L^2)} \leq C \abs z^{-\frac {n+1} 2}.
\]
\item Assume that the damping condition \eqref{hyp-damping} holds. Then there exists $C \geq 0$ such that for $z \in \C_+$ with ${\Re(z)} \geq C$ we have 
\[
\nr{\pppg x ^{-\d} \Rg^{n+1}(z) \pppg x ^{-\d}}_{\Lc(L^2)} \leq C \abs z^{-\frac {(n+1) \bb} 2}.
\]
\end{enumerate}
\end{theorem}

In order to prove the uniform estimates of Theorems \ref{th-inter-freq}, \ref{th-low-freq} and \ref{th-high-freq} we use the commutators method of Mourre (see \cite{mourre81}, see also \cite{amrein} and references therein for an overview of the subject). The method has been generalized to the dissipative setting in \cite{art-mourre}, then in \cite{boucletr14} for the estimates of the derivatives of the resolvent and finally in \cite{art-mourre-formes} for a dissipative perturbation in the sense of forms. Here the dissipative perturbation $\Ba$ is well defined as an operator on $L^2$ relatively bounded with respect to the self-adjoint part $\Pg$. However, for $d \in \{ 3,4\}$ the rescaled version of the dissipative part which we are going to use for low frequencies will be uniformly bounded as an operator in $\Lc(H^1,H\inv)$ but not in $\Lc(H^2,L^2)$, so we will have to see $\Hg$ as a dissipative perturbation of $\Pg$ in the sense of forms. See Remark \ref{rem-mourre-formes}.\\

Let us come back to the statement of Theorem \ref{th-smoothing-effect}. To prove this theorem we will use in particular the resolvent estimates of Theorem \ref{th-high-freq}, which in turn rely on the damping assumption \eqref{hyp-damping}. These estimates and hence the smoothing effect we obtain are optimal (in the sense that they are as good as in the self-adjoint case with the non-trapping condition) when $\a \geq 1$. However with a weaker dissipation ($\a < 1$) we can obtain (weaker) resolvent estimates and a (weaker) smoothing effect. Similarly, it is possible to prove high frequency resolvent estimates weaker than those of Theorem \ref{th-high-freq} without the damping condition. We have already mentioned \cite{burq04} in the self-adjoint case and \cite{alouikv13} in the dissipative setting, where only a few hyperbolic classical trajectories deny the assumption (in these cases the high-frequency resolvent estimates are of size $\ln\abs z / \sqrt{\abs z}$, which gives a gain of $\frac 12 - \e$ 
derivative). We do not prove resolvent estimates without damping condition in this paper, but we emphasize this fact with a more general version of Theorem \ref{th-smoothing-effect} (for self-adjoint operators, we mention the result of \cite{thomann10} which give a relation between the smoothing effect and the decay of the spectral projections).

\begin{theorem} \label{th-smoothing-bis}
Let $\bbb \in [0,2]$. Assume that there exists $C \geq 0$ such that for all $z \in \C_+$ we have 
\begin{equation} \label{hyp-res-estim}
\nr{\pppg x^{-1} \Rg(z) \pppg x^{-1}}_{\Lc(L^2)} \leq  C \pppg z^{-\frac \bbb 2}.
\end{equation}
Then for all $u_0 \in L^2$ we have 
\[
\int_0^{+\infty} \nr{\pppg x \inv \pppg {D}^{\frac {\bbb} 2} e^{-itH} u_0}^2_{L^2} dt \leq C \nr{u_0}_{L^2}^2.
\]
\end{theorem}

It is classical in the self-adjoint setting to prove the smoothing effect from resolvent estimates by means of the theory of relatively smooth operators in the sense of Kato (see \cite{kato66,rs4}). Other ideas have been used for dissipative operators (see \cite{alouikv13} and \cite{alouikr}). However the theory of Kato can also be used in this context (see \cite{art-mourre} and \cite{art-diss-schrodinger-guide}). We will follow this idea to prove Theorem \ref{th-smoothing-bis} and hence Theorem \ref{th-smoothing-effect}.\\

This paper is organized as follows. In Section \ref{sec-dissipative} we recall all the abstract properties we need concerning dissipative operators (including the statement of the Mourre method). In Section \ref{sec-inter-freq} we prove Theorem \ref{th-inter-freq}. In Section \ref{sec-low-freq} we deal with low frequencies. We first prove Theorem \ref{th-low-freq} for a small perturbation of the free laplacian in Section \ref{sec-low-freq-1} and then in the general setting in Section \ref{sec-low-freq-2}. Theorem \ref{th-low-freq-sharp} is proved in Section \ref{sec-low-freq-sharp}. In Section \ref{sec-high-freq} we prove Theorem \ref{th-high-freq} concerning the high-frequency resolvent estimates. Finally we turn to the time dependant problem: we prove Theorem \ref{th-loc-decay} in Section \ref{sec-loc-decay} and Theorems \ref{th-smoothing-bis} and \ref{th-smoothing-effect} in Section \ref{sec-smoothing-effect}.

\section{Abstract properties for dissipative operators} \label{sec-dissipative}

In this section we recall some general properties about dissipative operators. In particular we give the version of the Mourre's method which we use in this paper.\\

Let $\Hc$ be a Hilbert space. An operator $H$ with domain $\Dom(H)$ on $\Hc$ is said to be dissipative (respectively accretive) if 
\[
\forall \f \in \Dom(H), \quad \Im \innp{H\f} \f_\Hc \leq 0 \quad \big(\text{respectively } \Re \innp{H\f} \f_\Hc \geq 0 \big) .
\]
Moreover $H$ is said to be maximal dissipative (respectively maximal accretive) if it has no other dissipative (respectively accretive) extension than itself. Notice that $H$ is (maximal) dissipative if and only if $iH$ is (maximal) accretive. 
We recall that a dissipative operator $H$ is maximal dissipative if and only if there exists $z \in \C_+$ such that the operator $(H-z)$ has a bounded inverse on $\Hc$. In this case any $z \in \C_+$ belongs to the resolvent set of $H$ and 
\begin{equation} \label{estim-res-trivial}
\nr{(H-z)\inv}_{\Lc(\Hc)} \leq \frac 1 {\Im(z)}.
\end{equation}
According to the Hille-Yosida theorem this implies in particular that $-iH$ generates a contractions semi-group, and then for all $u_0 \in \Dom(H)$ the function $u : t \mapsto e^{-itH} u_0$ belongs to $C^0(\R_+,\Dom(H)) \cap C^1(\R_+,\Hc)$ and is the unique solution for the problem 
\[
\begin{cases} -i\partial_t u + Hu = 0, \quad \forall t > 0,\\
u(0) = u_0.
\end{cases}
\]
Moreover we have 
\begin{equation*}
\forall t \geq 0, \quad \nr{u(t)}_\Hc \leq \nr{u_0}_\Hc.
\end{equation*}

\begin{remark} \label{rem-dissipative-accretive}
Assume that $H$ is both dissipative and accretive. Then it is maximal dissipative if and only if it is maximal accretive. Indeed both properties are equivalent to the fact that $(H-(-1+i))$ has a bounded inverse on $\Hc$. Moreover in this case for $z \in \C$ with $\Im(z) > 0$ or $\Re(z) < 0$ we have 
\[
\nr{(H-z)\inv}_{\Lc(\Hc)} \leq \frac 1 { \max \big(\Im (z) , -\Re(z)\big)}.
\]
\end{remark}

\begin{proposition} \label{prop-max-diss-acc}
The operator $\Hg$ defined by \eqref{def-Ba-H} is maximal dissipative and maximal accretive on $L^2_w$.
\end{proposition}

\begin{proof}
The operator $\Pg$ and $\Ba$ are self-adjoint and non-negative, so that $\Hg = \Pg - i \Ba$ is dissipative and accretive. Let $\f \in H^2_w = \Dom(\Pg)$. By interpolation there exists $C \geq 0$ (which only depends on $a$ and $\a$) such that for any $\e > 0$
\[
\nr{\Ba \f}_{L^2_w} \leq C \nr{\f}_{H^\a_w} \leq C \nr{\f}_{H^2_w}^{\frac \a 2} \nr{\f}_{L^2_w}^{1-\frac \a 2} \leq \frac {\a \e C} 2 \nr{\f}_{H^2_w} + C \left( 1 - \frac \a 2 \right) \e ^{-\frac {\a}{2-\a} } \nr{\f}_{L^2_w}.
\]
With $\e > 0$ small enough we obtain that the dissipative operator $-i\Ba$ is relatively bounded with respect to $\Pg$ with relative bound less that 1. According to \cite[Lemma 2.1]{art-mourre} this proves that $\Hg$ is maximal dissipative in $L^2_w$. By Remark \ref{rem-dissipative-accretive}, $\Hg$ is also maximal accretive.
\end{proof}

According to Proposition \ref{prop-max-diss-acc} the estimate of Remark \ref{rem-dissipative-accretive} holds for $\Hg$ in $\Lc(L^2_w)$, and hence in $\Lc(L^2)$ up to a multiplicative constant. As already mentioned, the difficulties in Theorems \ref{th-inter-freq}, \ref{th-low-freq} and \ref{th-high-freq} come from the behavior of the resolvent $\Rg(z)$ when the spectral parameter $z \in \C_+$ approaches the non-negative real axis. For this we are going to use a dissipative version of the Mourre method, which we recall now.\\

Let $q_0$ be a quadratic form closed, densely defined, symmetric and bounded from below. We set $\Kc = \Dom(q_0)$. Let $q_\Th$ be another symmetric form on $\Hc$, non-negative and $q_0$-bounded. Let $q = q_0 - i q_\Th$ and let $H$ be the corresponding maximal dissipative operator (see Proposition 2.2 in \cite{art-mourre-formes}). We denote by $\tilde H : \Kc \to \Kc^*$ the operator which satisfies $q(\f,\p) = \innp{\tilde H \f} \p _{\Kc^*,\Kc}$ for all $\f,\p \in \Kc$. Similarly we denote by $\tilde H_0$ and $\Th$ the operators in $\Lc(\Kc,\Kc^*)$ which correspond to the forms $q_0$ and $q_\Th$, respectively. By the Lax-Milgram Theorem, the operator $(\tilde H -z)$ has a bounded inverse in $\Lc(\Kc^*,\Kc)$ for all $z \in \C_+$. Moreover for $\f \in \Hc$ we have $(H-z)\inv \f = (\tilde H - z)\inv \f$.

\begin{definition} \label{def-mourre}
Let $A$ be a self-adjoint operator on $\Hc$ and $N \in \N^*$. We say that $A$ is a conjugate operator (in the sense of forms) to $H$ on the interval $J$, up to order $N$, and with bounds $\a_0 \in ]0,1]$, $\b \geq 0$ and $\Upsilon_N \geq 0$ if the following conditions are satisfied:
\begin{enumerate} [(i)]
\item \label{item-Kc-invariant}
The form domain $\Kc$ is left invariant by $e^{-itA}$ for all $t \in \R$. We denote by $\Ec$ the domain of the generator of $\restr{e^{-itA}}{\Kc}$.
\item The commutators $B^0 = [\tilde H_0,iA]$ and $B_1 = [\tilde H,iA]$, \emph{a priori} defined as operators in $\Lc(\Ec,\Ec^*)$, extend to operators in $\Lc(\Kc,\Kc^*)$. Then for all $n \in \Ii 1 N$ the operator $[B_{n},iA]$ defined (inductively) in $\Lc(\Ec,\Ec^*)$ extends to an operator in $\Lc(\Kc,\Kc^*)$, which we denote by $B_{n+1}$. 
\item We have 
\[
\nr {B} \leq \sqrt \a_0 \Upsilon_N , \quad \nr {B + \b \Th} \nr{B_0} \leq \a_0 \Upsilon_N , \quad \nr{[B,A]} + \b \nr{[\Th,A]} \leq \a_0 \Upsilon_N
\]
and
\[
\sum_{n=2}^{N+1} {\nr{B_n}_{\Lc(\Kc,\Kc^*)}}  \leq \a_0 \Upsilon_N ,
\]
where all the norms are in $\Lc(\Kc,\Kc^*)$.
\item We have 
\begin{equation} \label{hyp-mourre}
\1 J (H_0) (B_0 + \b \Th) \1 J (H_0) \geq \a_0 \1 J (H_0).
\end{equation}
\end{enumerate}
\end{definition}

Theorem 5.5 of \cite{art-mourre-formes} in the particular case where all the inserted factors are equal to $\Id_\Hc$ gives the following abstract resolvent estimates:

\begin{theorem} \label{th-mourre}
Suppose that the self-adjoint operator $A$ is conjugate to the maximal dissipative operator $H$ on $J$ up to order $N \geq 2$ with bounds $(\a,\b,\Upsilon_N)$. Let $n \in \Ii 1 {N}$. Let $I \subset \mathring J$ be a compact interval. Let $\d > n- \frac 12$. Then there exists $c \geq 0$ which only depends on $J$, $I$, $\d$, $\b$ and $\Upsilon_N$ and such that for all $z \in \C_{I,+}$ we have
\[
\nr{\pppg A^{-\d} (H-z)^{-n} \pppg A^{-\d}}_{\Lc(\Hc)} \leq \frac c {\a_0 ^n}.
\]
\end{theorem}

We finish this general section with the so-called quadratic estimates. The following result is a consequence of Proposition 4.4 in \cite{art-mourre-formes}:

\begin{proposition} \label{prop-quad-estim}
Let $T \in \Lc(\Kc,\Hc)$ be such that $T^*T \leq q_\Th$ in the sense of forms on $\Kc$. Let $Q \in \Lc(\Hc,\Kc^*)$. Then for all $z \in \C_+$ we have 
\[
\nr{T (\tilde H-z)\inv Q}_{\Lc(\Hc)} \leq \nr{Q^* (\tilde H-z)\inv Q}_{\Lc(\Hc)} ^{\frac 12}.
\]
\end{proposition}

Applied with $Q = T^*$, this proposition gives the following particular case:

\begin{corollary} \label{cor-quad-estim}
Let $T$ be as in Proposition \ref{prop-quad-estim}. Then for all $z \in \C_+$ we have 
\[
\nr{T (\tilde H -z)\inv T^*}_{\Lc(\Hc)} \leq 1.
\]
\end{corollary}

We are going to use all these results with the forms $q_0 : \f \mapsto  \innp{P\f}\f$ and $q_\Th : \f \mapsto  \innp{\Ba \f}\f$ defined on $\Kc = H^1(\R^d)$.

\section{Intermediate frequency estimates} \label{sec-inter-freq}

In this section we prove Theorem \ref{th-inter-freq}. For this we will apply Theorem \ref{th-mourre} with the generator of dilations as the conjugate operator. Let 
\[
 A = - \frac i2 ( x \cdot \nabla + \nabla \cdot x) = - i \, (x \cdot \nabla) - \frac {id}2 .
\]
We recall in the following proposition the main properties of $A$ we are going to use in this paper:

\begin{proposition} \label{prop-A}
\begin{enumerate}[(i)]
\item For $\th \in \R$, $u \in \Sc$ and $x \in \R^d$ we have
\begin{equation*} 
 (e^{i\th A} u) (x) = e^{\frac {d\th}2} u ( e^\th x).
\end{equation*}
\item For $j \in \Ii 1 d$ and $\g \in C^\infty(\R^d)$ we have on $\Sc$:
\begin{equation*} 
 [\partial_j , i A] = \partial_j \quad \text{and} \quad [\g,iA] = -(x \cdot \nabla) \g.
\end{equation*}
 \item For $p \in [1,+\infty]$, $\th \in \R$ and $u \in \Sc$ we have
\[
 \nr{e^{i\th A} u}_{L^p} = e^{\th \left( \frac d 2 - \frac d p\right)} \nr u_{L^p}.
\]
\end{enumerate}
\end{proposition}

Now we give a proof of Theorem \ref{th-inter-freq}:

\begin{proof} [Proof of Theorem \ref{th-inter-freq}]
Let $E > 0$. We check that the generator of dilations $A$ is a conjugate operator for $\Hg$ on a neighborhood $J$ of $E$ in the sense of Definition \ref{def-mourre}. The form domain of $\Hg$ is the Sobolev space $H^1(\R^d)$. According to Proposition \ref{prop-A}, it is left invariant by the dilation $e^{-itA}$ for any $t \in \R$. By pseudo-differential calculus we can see that the commutators $[\Pg,iA]$, $[[\Pg,iA],iA]$, $[\Ba,iA]$ and $[[\Ba,iA],iA]$ define operators in $\Lc(H^2,L^2)$, hence in $\Lc(L^2,H^{-2})$ by duality, and in $\Lc(H^1,H\inv)$ by interpolation\footnote{ In fact we can also compute these commutators explicitely with Proposition \ref{prop-A}, except for the commutators of $\pppg D^\a$ with $A$: for this we can write $\pppg D^\a = (1-\D)^2 \times {(1-\D)^{-2}} \pppg {-\D}^{\frac \a 2}$ and use the Helffer-Sjöstrand formula for the second factor (see \cite{dimassis,davies95}).}. Since we need estimates for a single operator, we do not have to worry about the estimates of the third 
assumption. We only have to choose $\Upsilon_N$ large enough. Finally we use the usual trick for the main assumption. For $\s > 0$ we set $J_\s = [E-\s,E+\s]$. We have 
\begin{align*}
\1{J_\s}(\Pg)  [\Pg,iA]  \1{J_\s}(\Pg)
& =  \1{J_\s}(\Pg) \, 2\Pg \,  \1{J_\s}(\Pg) +  \tilde W \1{J_\s}(\Pg)  \\
& \geq 2(E-\s)  \1{J_\s}(\Pg)  +  \tilde W \1{J_\s}(\Pg) 
\end{align*}
where 
\[
\tilde W =  \1{J_\s}(\Pg) \left(\divg  \big((x\cdot \nabla) G(x) \big)\nabla - \sum_{j=1}^d (x\cdot \nabla) b_j D_j - \sum_{j=1}^d b_j D_j \right) 
\]
is a compact operator. Since $E > 0$ is not an eigenvalue of $\Pg$ (see \cite{kocht06}) the operator $\1{J_\s}(\Pg)$ goes strongly to 0 when $\s$ goes to 0. Then for $\s$ small enough we have 
\[
\1{J_\s}(\Pg)  [\Pg,iA] \1{J_\s}(\Pg) \geq E  \1{J_\s}(\Pg).
\]
Thus we can apply Theorem \ref{th-mourre}, which gives Theorem \ref{th-inter-freq} for $\Re(z) \in J_\s$ and with weights $\pppg A^{-\d}$. By compactness of $K \subset \C^*$ and the easy estimate of Remark \ref{rem-dissipative-accretive}, we have a uniform estimate for all $z \in \C_+ \cap K$. It remains to replace $\pppg A^{-\d}$ by $\pppg x^{-\d}$. For this we use the resolvent identity 
\[
\Rg(z) = \Rg(i) + (z-i) \Rg(i) \Rg(z) = \Rg(i) + (z-i)\Rg(z) \Rg(i) 
\]
to prove by induction on $m \in \N^*$ that $\Rg^{n+1}(z)$ can be written as a sum of terms of the form $(z-i)^\b \Rg^{n+1+\b}(i)$ with $\b \in \N$ or 
\[
(z-i)^{2m - n- 1 + \n} \Rg^m(i) \Rg^{\n}(z) \Rg^m(i)
\]
with $\max(1,n+1 - 2m) \leq \n \leq n+1$. On the one hand $\Rg^{n+1+\b}(i)$ is uniformly bounded in $\Lc(L^2)$ and on the other hand 
\begin{eqnarray*}
\lefteqn{\nr{\pppg x^{-\d} \Rg^m(i) \Rg^{\n}(z) \Rg^m(i) \pppg x^{-\d}}}\\
&&\leq \nr{\pppg x^{-\d} \Rg^m(i) \pppg A^\d}\nr{\pppg A^{-\d} \Rg^{\n}(z) \pppg A^{-\d}} \nr{\pppg A^\d \Rg^m(i) \pppg x^{-\d}}.
\end{eqnarray*}
The first and third factors are bounded by pseudo-differential calculus if $m$ is large enough and the second has been estimated uniformly by the Mourre method. This concludes the proof of Theorem \ref{th-inter-freq}.
\end{proof}

\section{Low frequency estimates} \label{sec-low-freq}

In this section we prove Theorems \ref{th-low-freq} and \ref{th-low-freq-sharp}. As in \cite{bouclet11, boucletr14}, the proof of Theorem \ref{th-low-freq} is based on a scaling argument for a small perturbation of the free Laplacian (see Section \ref{sec-low-freq-1}), and then on a perturbation argument to deal with the general case (see Section \ref{sec-low-freq-2}). Theorem \ref{th-low-freq-sharp} is proved in Section \ref{sec-low-freq-sharp}.\\

Let $\h \in C_0^\infty(\R^d)$ be equal to 1 on a neighborhood of 0. For $\y \in ]0,1]$ we set $\h_\y : x \mapsto \h(\y x)$. Then for $\yy \in ]0,1]$ we set $G_{\yy}(x) = \h_{\yy}(x) I_d + (1-\h_{\yy}(x)) G(x)$,
\begin{equation} \label{def-Pii-Pcc}
\Pii = -\divg G_{\yy}(x) \nabla \quad \text{and} \quad  \Pcc = \Pg - \Pii = -\divg \big(\h_{\yy}(x) (G(x)-I_d)\big) \nabla .
\end{equation}
For the dissipative part we set
\begin{equation} \label{def-aii}
\Bai =   a(1-\h_{\yyy}) \pppg{D}^\a a +  a \h_{\yyy} \pppg{D}^\a a(1-\h_{\yyy})  
\end{equation}
and
\begin{equation*} 
\Bac = \Ba - \Bai =  a  \h_{\yyy} \pppg{D}^\a a\h_{\yyy},
\end{equation*}
where $\yyy \in ]0,1]$. Finally, for the full operator we define 
\[
\Hii = \Pii - i \Bai
\quad 
\text{
and 
}
\quad
\Ri (z) = (\Hii - z)\inv,
\]
where $\ybar = (\yy,\yyy) \in ]0,1]^2$.

\subsection{Low frequency estimates for a small perturbation of the Laplacian} \label{sec-low-freq-1}

In this paragraph we prove Theorem \ref{th-low-freq} with $\Rg(z)$ replaced by $\Ri (z)$. Then in Section \ref{sec-low-freq-2} we will add the contributions of $\Pcc$, $W$ and $\Bac$.\\

The proof relies on a scaling argument. To this purpose we use for $z \in \C^*$ the operator
\[
\Thz = \exp \left( \frac {i \ln \abs z}2 A \right).
\]
For $u \in \Sc$ and $x \in \R^d$ we have $(\Th_z u)(x) = \abs z^{d/4} u \big(\abs z^{1/2} x\big)$. According to Proposition \ref{prop-A} we have for $p \in [1,+\infty]$
\begin{equation} \label{estim-dil-Lp}
\nr{\Thz}_{\Lc(L^p)} = \abs z^{\frac d 4 - \frac d {2p}}.
\end{equation}
For a function $u$ on $\R^d$ and $z \in \C^*$ we denote by $u_z$ the function 
\[
u_z : x \mapsto u \left( \frac x {\sqrt{\abs z}} \right).
\]
Compared to the scaling for the wave equation we are using the parameter $\sqrt {\abs z}$ instead of $\abs z$.\\

Now we introduce the rescaled versions of our operators:
\[
\tHii = \frac 1 {\abs z} \Thz \inv \Hii \Thz = \tPii -   i \tBai
\]
where $\tPii = - \divg G_{\yy,{z}}(x) \nabla$ and
\[
\tBai = \frac 1 {\abs z} \big((1-\h_{\yyy})a \big)_{z} \big(1 - \abs z \D \big)^{\frac \a 2}  a_{z} + \frac 1 {\abs z} (\h_{\yyy} a)_z \big(1 - \abs z \D \big)^{\frac \a 2} \big((1-\h_{\yyy})a \big)_{z}.
\]
Then for $\z \in \C_+$ we set
$
\Riz (\z) = \big( \tHii - \z \big) \inv,
$
so that with the notation $\hat z = z / \abs z$ we have for $z \in \C_+$
\[
\Ri (z) = \frac 1 {\abs z} \Thz \Riz (\hat z) \Thz\inv.
\]

Our analysis of the rescaled operators is based on the fact that if a function $\vf$ decays like $\pppg x^{-\n - \frac \rho 2}$ (recall that $\rho >0$ is fixed by \eqref{hyp-long-range} and \eqref{hyp-a-short-range}) then the multiplication by the rescaled function $\vf_{\l}$ behaves like an differential operator of order $\n$ for low frequencies, in the sense that it is of size $\l^\n$ as an operator from $H^s$ to $H^{s-\n}$. Since this observation relies on the Sobolev embeddings, there is however a restriction in the choice of $\n$ and $s$. For $\s \in \R$, let $\symb^{-\s}(\R^d)$ be the set of functions $\vf \in C^\infty(\R^d)$ such that
\[
\abs{\partial ^\b \vf(x)} \lesssim \pppg x^{-\s -\abs \b}.
\]
For $\n \geq 0$, $N \in \N$ and $\vf \in \Sc^{-\n - \frac \rho 2}(\R^d)$ we set
\[
\nr{\vf}_{\n,N} = \sup_{\abs \b \leq \kk +1} \sum_{0\leq m \leq N} \sup _{x\in \R^d} \abs{\pppg x^{\n + \frac \rho 2 + \abs \b} \big( \partial^\b (x\cdot \nabla)^m \vf \big) (x)}.
\]
We recall that the integer $\kk$ was defined in \eqref{def-talpha-kappa}. The following result is Proposition 7.2 in \cite{boucletr14}:

\begin{proposition} \label{prop-dec-sob}
Let $\n \in \big[ 0, \frac d 2 \big[$ and $s \in \big] -\frac d 2, \frac d 2\big[$ be such that $s -\n \in \big] -\frac d 2, \frac d 2\big[$. Then there exists $C \geq 0$ such that for $\vf \in \symb^{-\n-\frac \rho 2}(\R^d)$, $u \in H^s$ and $\l > 0$ we have
\[
 \nr{\vf_\l u}_{\dot H^{s-\n}} \leq C \l ^\n \nr \vf_{\n,0} \nr u _{\dot H^s}
\]
and
\[
\nr{\vf_\l u}_{H^{s-\n}} \leq C \l ^\n \nr \vf_{\n,0} \nr u _{H^s}.
\]
\end{proposition}

The interest of replacing $G(x)$ by $G_{\yy}(x)$ and $a$ by $a(1-\h_{\yyy})$ in the definition of $\Hii$ is that for all $N \in \N$ we have 
\begin{equation} \label{def-estim-Nc}
\begin{aligned}
\Nc_{\ybar,N}
& : = \sum_{j,k = 1}^d \nr{G_{\yy,j,k}(x) - \d_{j,k}}_{0,N} + \nr{(1-\h_{\yyy}) a}_{1,N} \left(\nr{ a}_{1,N}+\nr{\h_{\yyy} a}_{1,N} \right)\\
&  = \bigo \ybar 0 \big(\abs{\ybar}^{\rho / 2} \big).
\end{aligned}
\end{equation}
Thus this quantity is as small as we wish if we choose $\yy$ and $\yyy$ small enough.\\

Given two operators $T$ and $S$ we set $\ad_T^0(S) = S$, $\ad_T(S) = \ad_T^1(S) = [S,T]$ and then for $m \geq 2$: $\ad_T^m(S) = [\ad_T^{m-1}(S),S]$. For $\m = (\m_1,\dots,\m_d) \in \N^d$ we set 
\[
\ad_x^\m := \ad_{x_{1}}^{\m_1} \dots \ad_{x_{d}}^{\m_d}.
\]

At the beginning of the section we said that $\Hii$ has to be close to the free Laplacian. What we need precisely is the following result:

\begin{proposition} \label{prop-Pii-Hii-D}
Let $\m \in \N^d$, $m\in \N$, $\e_0 > 0$ and $s \in \R$. There exists $\y_0 \in ]0,1]$ such that for $\ybar =(\yy,\yyy) \in ]0,\y_0]^2$ the following statements hold:
\begin{enumerate}[(i)]
\item If $s \in \big]-\frac d 2, \frac d 2\big[$ then for $z \in \C_+$ with $\abs z \leq 1$ we have 
\[
\nr{\ad_x^\m \ad_{A}^m \big( \tPii +\D)}_{\Lc(H^{s+1},H^{s-1})} \leq \e_0.
\]
\item If $s \in \big]-\frac d 2+1, \frac d 2-1\big[$ then we also have
\[
\nr{\ad_x^\m \ad_{A}^m \tBai}_{\Lc(H^{s+1},H^{s-1})} \leq \e_0.
\]
\item For $u \in H^2$ we have 
\[
\frac 12 \nr u _{\dot H^2} \leq \nr{\Pii u}_{L^2} \leq 2 \nr u _{\dot H^2}.
\]
\end{enumerate}
\end{proposition}

\begin{proof}
The first statement is the same as for the wave equation. See Proposition 7.6 in \cite{boucletr14}. In particular with $s = 1$, $\abs z = 1$ and $\e_0 = \frac 12$ we obtain the last statement. It remains to prove (ii). Let $D_z = \sqrt {\abs z} D$. We write 
\[
((1-\h_{\yyy}) a)_z  \pppg {D_z}^{\a}a_z = ((1-\h_{\yyy}) a)_z (-\abs z \D + 1) \pppg {D_z}^{\a - 2} a_z.
\]
Then $\ad_{x}^{\m} \ad_A^{m} \Big(((1-\h_{\yyy}) a)_z  \pppg {D_z}^{\a}a_z \Big)$ can be written as a sum of terms of the form 
\[
\ad_{x}^{\m_1} \ad_A^{m_1}\big( ((1-\h_{\yyy}) a)_z \big) \, \ad_{x}^{\m_2} \ad_A^{m_2} \big(-\abs z \D + 1 \big) \, \ad_{x}^{\m_3} \ad_A^{m_3} \big(\pppg {D_z}^{\a - 2} \big) \, \ad_{x}^{\m_4} \ad_A^{m_4} (a_z) ,
\]
where $\m_1,\m_2,\m_3,\m_4 \in \N^d$ and $m_1,m_2,m_3,m_4 \in \Ii 0 m$ are such that $\m_1 + \m_2 + \m_3 + \m_4 = \m$ and $m_1 + m_2 + m_3 +m_4 = m$. Let $\g \in [0,1]$. According to Proposition \ref{prop-dec-sob} we have for $z \in \C_+$
\begin{equation} \label{estim-m1}
\nr{\ad_x^{\m_1} \ad_{A}^{m_1}\big((1-\h_{\yyy})a \big)_{z}}_{\Lc(H^{s-1+\g},H^{s-1})} \lesssim \yyy^{1+ \frac \rho  2 -\g} \abs z^{\frac \g 2}
\end{equation}
and
\begin{equation} \label{estim-m4}
 \nr{\ad_x^{\m_4} \ad_{A}^{m_4} a_z }_{\Lc(H^{s+1},H^{s+1-\g})}\lesssim  \abs z^{\frac \g 2}.
\end{equation}
To estimate $\ad_{x}^{\m_3} \ad_A^{m_3} \big(\pppg {D_z}^{\a - 2} \big)$ we use the Helffer-Sj\"ostrand formula (see \cite{dimassis,davies95}). We can check that for $\z \in \C \setminus \R$ we have 
\[
\nr{\ad_{x}^{\m_3} \ad_A^{m_3} \big(-\abs z \D - \z \big) \inv }_{\Lc(H^{s+1-\g})} \lesssim \frac {\pppg \z^{\abs {\m_3} + m_3}} {\abs {\Im(\z)}^{\abs {\m_3} + m_3 +1}}.
\]
Let $f : \t \mapsto (\t + 1)^{\frac {\a -2}2}$. Let $\vf \in C_0^\infty(\R,[0,1])$ be supported in $[-2,2]$ and equal to 1 on $[-1,1]$. For $m > \abs {\m_3} + m_3 +1$ and $\z = x + i y$ we set 
\[
 \tilde f_{m} (\z) =  \vf \left(  \frac {y} {\pppg x} \right) \sum_{k=0}^m f^{(k)} (x) \frac {(iy)^k}{k!}.
\]
We have
\begin{align*}
\abs{\frac {\partial \tilde f_m}{\partial \bar \z} (\z)} \leq \1 {\singl{\pppg {x} \leq \abs{y} \leq 2 \pppg {x}} } (\z) \pppg {x}^{-1+ \frac {\a - 2}2} + \1 {\singl{ \abs{y} \leq 2 \pppg{x}}} (\z) \abs{y}^m \pppg {x}^{-m-1+ \frac {\a - 2}2},
\end{align*}
so we can write
\[
(-\abs z \D + 1)^{\frac {\a -2}2} = \frac 1 \pi \int_{\z = x+iy \in \C} \frac {\partial \tilde f_m}{\partial \bar z}(\z) (-\abs z \D - \z)\inv \, dx \, dy.
\]
Then we can check that 
\begin{equation} \label{estim-m3}
\nr{\ad_{x}^{\m_3} \ad_A^{m_3} \pppg {D_z}^{\a -2} }_{\Lc(H^{s+1-\g})} \lesssim 1.
\end{equation}
It remains to estimate 
\begin{equation} \label{decomp-m2}
\ad_{x}^{\m_2} \ad_A^{m_2} \big(\pppg {D_z}^{2}\big) = -\abs z \ad_{x}^{\m_2} \ad_A^{m_2} \D  + \ad_{x}^{\m_2} \ad_A^{m_2} (1).
\end{equation}
We have $\nr{\abs z \ad_{x}^{\m_2} \ad_A^{m_2} \D } \lesssim \abs z$ in $\Lc(H^{s+1},H^{s-1})$ so with \eqref{estim-m1}, \eqref{estim-m4} and \eqref{estim-m3} applied with $\g = 0$ we obtain in $\Lc(H^{s+1},H^{s-1})$
\begin{multline}
\nr{\ad_{x}^{\m_1} \ad_A^{m_1}\big( ((1-\h_{\yyy}) a)_z \big)  \, \ad_{x}^{\m_2} \ad_A^{m_2} \big(-\abs z \D\big) \, \ad_{x}^{\m_3} \ad_A^{m_3} \big(\pppg {D_z}^{\a - 2} \big) \, \ad_{x}^{\m_4} \ad_A^{m_4}(a_z) }\\
 \lesssim \abs z \yyy^{1 + \frac \rho 2}.
\end{multline}
If $\abs {\m_2} = m_2 = 0$ we also have to consider the second term in \eqref{decomp-m2}. For this we apply \eqref{estim-m1}, \eqref{estim-m4} and \eqref{estim-m3} with $\g = 1$, which gives 
\[
\nr{\ad_{x}^{\m_1} \ad_A^{m_1} \big( ((1-\h_{\yyy}) a)_z  \, \ad_{x}^{\m_3} \ad_A^{m_3} \big(\pppg {D_z}^{\a - 2} \big) \, \ad_{x}^{\m_4} \ad_A^{m_4}(a_z)  \big)}_{\Lc(H^{s+1},H^{s-1})} \lesssim \abs z \yyy^{\frac \rho 2}.
\]
Thus we have proved that $\ad_{x}^{\m} \ad_A^{m} \Big(((1-\h_{\yyy}) a)_z \pppg {D_z}^{\a} a_z \Big)$ is of size $O(\abs z \yyy^{\frac \rho 2})$ in $\Lc(H^{s+1},H^{s-1})$. We proceed similarly for $\ad_{x}^{\m} \ad_A^{m} \Big((\h_{\yyy} a)_z \pppg {D_z}^{\a} ((1-\h_{\yyy}) a)_z \Big)$, and the statement follows.
\end{proof}

\begin{remark} \label{rem-low-dim}
If $d \geq 5$ we can replace $\Pii$ by $\Hii$ in the last statement of Proposition \ref{prop-Pii-Hii-D}. This is not the case for $d \in \{ 3, 4\}$. This is due to the fact that $s = 1$ does not belong to $\big]-\frac d 2+1, \frac d 2-1\big[$ and hence $\tBai$ is not small in $\Lc(\dot H^2,L^2)$ in these cases.
\end{remark}

\begin{proposition} \label{prop-resolvante-regularisante}
Let $\m \in \N^d$, $m \in \N$ and $s \in \big] -\frac d 2 + 1 , \frac d 2 -1 \big[$. There exists $\y_0 \in ]0,1]$ such that the operator
\[
\ad_x^\m \ad_A^m \Riz (-1)
\]
is bounded as an operator from $H^{s-1}$ to $H^{s+1}$ uniformly in $z \in \C_+$ with $\abs z \leq 1$ and $\ybar =(\yy,\yyy) \in ]0,\y_0]^2$ .
\end{proposition}

\begin{proof}
The idea of the proof is the same as the proof of Proposition 7.9 in \cite{boucletr14}. We only have to be careful with the fact that the dissipative term has to be seen as an operator of order 2. However, with the smallness assumption on $a(1-\h_{\yyy})$, it is still a small perturbation of $-\D$, and we can proceed as for the wave resolvent. We also have to be careful with the restriction on $s$ which is stronger than for the wave equation. This is due to the analogous restriction in the second statement of Proposition \ref{prop-Pii-Hii-D}. We omit the details.
\end{proof}

\begin{proposition} \label{prop-gain-par-regularite}
\begin{enumerate}[(i)]
\item Let $s \in \big[ 0 , \frac d 2 \big[$, $\d > s$ and $m \in \N$ be such that $m\geq s$. Then there exist $\y_0 \in ]0,1]$ and $C \geq 0$ such that for $z \in \C_+$ with $\abs z \leq 1$ and $\ybar \in ]0,\y_0]^2$ we have 
\[
\nr{\pppg x^{-\d} \Thz \Riz^m (-1)  \Thz\inv \pppg x^{-\d}}_{\Lc(L^2)} \leq C \abs z^{s}.
\]

\item 
Let $s \in \big[ 0 , \frac d 2 \big[$, $\d > s$, and $m \in \N$ large enough (say $m \geq \d + \frac s 2 + 1$).  Then there exist $\y_0 \in ]0,1]$ and $C \geq 0$ such that for $z \in \C_+$ with $\abs z \leq 1$ and $\ybar \in ]0,\y_0]^2$ we have 
\[
\nr{\pppg x^{-\d} \Thz \Riz^m (-1) \pppg A^\d}_{\Lc(L^2)} \leq C \abs z^{\frac s 2}
\]
and 
\[
\nr{\pppg A^\d  \Riz^m (-1) \Thz\inv \pppg x^{-\d} }_{\Lc(L^2)} \leq C \abs z^{\frac s 2}.
\]
\end{enumerate}
\end{proposition}

\begin{proof}
According to Proposition \ref{prop-resolvante-regularisante} the operator $\Riz^m (-1) $ is bounded in $\Lc(H^{-s},H^s)$ uniformly in $z \in \C_+$ with $\abs z \leq 1$ and $\ybar$ close to (0,0). On the other hand, according to the Sobolev embedding $H^s \subset L^p$ for $p = \frac {2d}{d-2s}$, the fact that $\pppg x^{-\d}$ belongs to $\Lc(L^p,L^2)$ and \eqref{estim-dil-Lp} we have 
\[
\nr{\pppg x^{-\d} \Thz}_{\Lc(H^s,L^2)} \lesssim \nr{\Thz}_{\Lc(L^p)} \lesssim \abs{z}^{\frac s 2}.
\]
We similarly have 
\[
\nr{\Thz\inv \pppg x^{-\d}}_{\Lc(L^2,H^{-s})} \lesssim \abs{z}^{\frac s 2},
\]
and the first statement follows. For the second statement we use the same idea as in the proof of Proposition 7.11 in \cite{boucletr14}. We only prove the first estimate. For this we first remark that
\begin{align*}
\nr{\pppg x^{-\d} \Thz \big( 1 + \abs x^\d \big) }_{\Lc(H^s,L^2)}
& \leq \nr{\pppg x^{-\d} \Thz}_{\Lc(H^s,L^2)} + \nr{\pppg x^{-\d} \Thz  \abs x^\d  }_{\Lc(L^2)}\\
& \lesssim \nr{\Thz}_{\Lc(L^p)} + \abs z^{\frac \d 2} \nr{\pppg x^{-\d}   \abs x^\d \Thz }_{\Lc(L^2)}\\
& \lesssim \abs z^{\frac s 2}
\end{align*}
where, again, $p$ stands for $\frac {2d}{d-2s}$. Then it remains to prove that for all $\d \geq 0$ (we no longer need the assumption that $\d > s$), $m\geq \d + \frac s 2 + 1$ and $\m \in \N^d$ the operator 
\[
\pppg x^{-\d} \ad_x^\m \left( \Riz^m (-1) \pppg A^\d \right)
\]
is bounded in $\Lc(L^2,H^s)$ uniformly in $z \in \C_+$. With $\m = 0$ this will conclude the proof. By interpolation it is enough to consider the case where $\d$ is an integer and $m \geq \d + \frac s 2$ (we do not mean to be sharp with this assumption). We proceed by induction. The statement for $\d = 0$ is given by Proposition \ref{prop-resolvante-regularisante}. Now let $\d \in \N^*$. We have 
\[
 \Riz^m (-1) A^\d = \sum_{k=0}^\d C_\d^k \Riz^{m-1} (-1) A^{\d-k} \ad_A^k \big(\Riz (-1)\big).
\]
When $k \neq 0$ we can apply the inductive assumption to $\Riz^{m-1} (-1) A^{\d-k}$. With Proposition \ref{prop-resolvante-regularisante} we obtain that the contributions of the corresponding terms are uniformly bounded in $\Lc(L^2,H^s)$ as expected. It remains to consider the term corresponding to $k=0$. It is enough to consider 
\[
\Riz^{m-1} (-1) A^{\d-1} x_j D_j \Riz (-1) 
\]
for some $j \in \Ii 1 d$. The operator $ D_j \Riz (-1) $ and its commutators with powers of $x$ are uniformly bounded operators on $L^2$, and 
\[
\Riz^{m-1} (-1) A^{\d-1} x_j = x_j \Riz^{m-1} (-1) A^{\d-1} + \ad_{x_j} \big(\Riz^{m-1} (-1) A^{\d-1}\big).
\]
We conclude with the inductive assumption.
\end{proof}

\begin{proposition} \label{prop-mourre-low-freq}
Let $k \in \N$ and $\d > k + \frac 12$.  Then there exist $\y_0 \in ]0,1]$ and $C \geq 0$ such that for $z \in \C_+$ with $\abs z \leq 1$ and $\ybar =(\yy,\yyy) \in ]0,\y_0]^2$ we have 
\[
\nr{\pppg A^{-\d} \Riz^{k+1} (\hat z)\pppg A^{-\d}}_{\Lc(L^2)} \leq C.
\]
\end{proposition}

\begin{proof}
The estimate is clear when $\hat z$ is outside some neighborhood of 1. For $\hat z$ close to 1 we apply Theorem \ref{th-mourre} uniformly in $z$ with $A$ as a conjugate operator. We have already said that $e^{-itA}$ leaves $H^1$ invariant for all $t \in \R$. The assumptions (ii) and (iii) of Definition \ref{def-mourre} with $\a_0=1/2$ and $\b= 0$ are consequences of Proposition \ref{prop-Pii-Hii-D} applied with $s = 0$, $\m = 0$ and $m \in \N^*$. For $m \in \{0,1\}$, $z \in \C_+$ and $u \in \Sc$ we have 
\begin{eqnarray*}
\lefteqn{\abs{\innp{ \ad_{iA}^m (\tPii) u}{u}_{L^2} - 2^m \innp{-\D u}{u}_{L^2}}}\\
&& \leq \sum_{j,k= 1}^d \abs{\innp{(2-x\cdot \nabla)^m (G_{\yy,z,j,k} - \d_{j,k}) D_j u}{D_k u}_{L^2}}\\
&& \lesssim O(\yy^{\rho/2}) \nr{\nabla u}^2_{L^2},
\end{eqnarray*}
and hence
\[
[\tPii,iA] \geq \Big(2 - O\big(\yy^{\rho / 2} \big)\Big)(- \D)  \geq \Big(2 - O\big(\yy^{\rho/ 2} \big)\Big)  \tPii.
\]
Let $J = \big] \frac 12 , \frac 32 \big[$. After conjugation by $\1 J (\tPii)$ we obtain that if $\y_0$ is small enough then for all $\yy \in ]0,\y_0]$ and $z \in \C_+$ we have 
\[
\1 J (\tPii) [\tPii,iA] \1 J (\tPii) \geq \frac 12 \1 J (\tPii).
\]
Then Proposition \ref{prop-mourre-low-freq} follows from Theorem \ref{th-mourre}.
\end{proof}

\begin{remark} \label{rem-mourre-formes}
It is important to notice that we have estimated $[\Bai,iA]$ and $[[\Bai,iA],iA]$ in $\Lc(H^1,H\inv)$ and not in $\Lc(H^2,L^2)$. By pseudo-differential calculus, these two commutators define operators in $\Lc(H^2,L^2)$. But in low dimensions ($d \in \{3.4\}$) they can be estimated uniformly by Proposition \ref{prop-Pii-Hii-D} only in the sense of forms. This is why we need a form version of the dissipative Mourre method here.
\end{remark}

\begin{proposition} \label{prop-low-freq-petite-perturbation}
Let $\e > 0$ and $n \in \N$. Let $\d$ be as in the statement of Theorem \ref{th-low-freq}. Then there exist $\y_0 \in ]0,1]$, $C\geq 0$ and a neighborhood $\Uc$ of $0$ in $\C$ such that for $\ybar =(\yy,\yyy) \in ]0,\y_0]^2$ and $\b_l,\b_r \in \R_+$ with ${\b_l} + {\b_r} \leq 2$ we have 
\[
\nr{\pppg x^{-\d} \pppg D^{\b_l} \Ri^{n+1}(z) \pppg D^{\b_r} \pppg x ^{-\d}}_{\Lc(L^2)} \leq C \left( 1 + \abs z ^{\frac d 2 - \e - 1 - n} \right).
\]
\end{proposition}

\begin{remark}
Compared to the analogous result for the wave equation (see Theorem 1.3 in \cite{boucletr14}) there is no gain when we add a derivative. This is a consequence of the restriction on the Sobolev index $s$ in Proposition \ref{prop-resolvante-regularisante}, which is stronger than in Proposition 7.9 in \cite{boucletr14}.
\end{remark}

\begin{proof} [Proof of Proposition \ref{prop-low-freq-petite-perturbation}]
First assume that $n \geq 1$. By the resolvent identity we have
\begin{align*}
\pppg x^{-\d} \pppg D^{\b_l} \Ri^{n+1}(z) \pppg D^{\b_r} \pppg x^{-\d}
& = {\pppg x^{-\d} \pppg D^{\b_l} \Ri(-1) \pppg x^\d} \\
&\quad  \times \pppg x^{-\d} \left( \Ri^{n-1}(z) + 2 (1+z) \Ri^{n}(z) + (1+z)^2 R^{n+1}(z) \right) \pppg x^{-\d}\\
&\quad  \times \pppg x^\d \Ri(-1) \pppg D^{\b_r} \pppg x^{-\d}.
\end{align*}
The first and last factors are bounded on $L^2$ uniformly in $\ybar \in ]0,1]^2$ by pseudo-differential calculus, so it is enough to prove the statement without additionnal derivatives if $n \geq 1$. Since ${\b_l} +  {\b_r} \leq 2$ we have a similar argument for $n=0$.

We have
\[
\pppg x^{-\d} \Ri^{n+1} (z) \pppg x^{-\d} = \abs z^{-(n+1)} \pppg x^{-\d} \Thz \Riz^{n+1} (\hat z) \Thz \inv \pppg x^{-\d}.
\]
As in the proof of Theorem \ref{th-inter-freq} in Section \ref{sec-inter-freq} we can prove by induction on $m \in \N^*$ that $\Riz^{n+1} (\hat z)$ can be written as a sum of terms of the form 
\begin{equation} \label{decomp-tRizi}
(1+\hat z)^\b \Riz^{n+1+\b} (-1) \quad \text{or} \quad (1+\hat z)^{2m-n-1+\n} \Riz^m (-1) \Riz^\n (\hat z) \Riz^m (-1),
\end{equation}
where $\max(1,n+1-2m) \leq \n \leq n+1$ and $\b \in \N$. Let $s = \min\big(n+1,\frac d 2 - \e\big)$. For $\b \in \N$ we have $s \in \big[0,\frac d 2 \big[$, $n+1+\b \geq s$ and $\d > s$ so according to the first statement of Proposition \ref{prop-gain-par-regularite} we have 
\[
\abs z^{-(n+1)} \nr{\pppg x^{-\d} \Riz^{n+1+\b} (-1) \pppg x^{-\d}}_{\Lc(L^2)} \lesssim \abs z^{s-(n+1)} \lesssim 1 + \abs z^{\frac d 2 - \e - n- 1}.
\]
Now we consider the contributions of terms of the second kind in \eqref{decomp-tRizi}. We can assume that $m$ is large enough to apply the second statement of Proposition \ref{prop-gain-par-regularite}. We have $\d > \n - \frac 12$ so with proposition \ref{prop-mourre-low-freq} we get
\begin{eqnarray*}
\lefteqn{ \abs z^{-(n+1)} \nr{\pppg x^{-\d} \Thz \Riz^m (-1) \Riz^\n (\hat z) \Riz^m (-1) \Thz \inv \pppg x^{-\d}}}\\
&& \leq \abs z^{-(n+1)} \nr{\pppg x^{-\d} \Thz \Riz^m (-1) \pppg A^\d} \nr{\pppg A^{-\d} \Riz^\n (\hat z) \pppg A^{-\d}} \nr{\pppg A^\d  \Riz^m (-1) \Thz \inv \pppg x^{-\d}}\\
&&\lesssim \abs z^{s-(n+1)}  \lesssim 1 + \abs z^{\frac d 2 - \e - n- 1}.
\end{eqnarray*}
This concludes the proof.
\end{proof}

\subsection{Low frequency estimates for a general perturbation of the Laplacian} \label{sec-low-freq-2}

In this paragraph we use the estimates on $\Ri(z)$ to prove the same estimates for $\Rg(z)$. To this purpose we have to add the contributions of $\Pcc$ and $W$ in the self-adjoint part, and the contribution of $\Bac$ in the dissipative part.\\

For $\y_0, \yyy \in ]0,1]$ and $\ybar = (\y_0,\yyy)$ we set $\Kz = \Pcco + W$ and, for $\p \in C_0^\infty(\R)$,
\[
\Spz = \Kz \Ri (z) \p(\Pg).
\]
From Proposition \ref{prop-low-freq-petite-perturbation} we obtain the following result:

\begin{proposition} \label{prop-estim-S}
Let $\e > 0$, $n \in \N$ and $M \in \R$. Let $\p \in C_0^\infty(\R)$. Let $\d$ be as in the statement of Theorem \ref{th-low-freq}. Let $\y_0 \in ]0,1]$ be given by Proposition \ref{prop-low-freq-petite-perturbation}. Then there exists $C \geq 0$ and a neighborhood $\Uc$ of $0$ in $\C$ such that for $\yyy \in ]0,\y_0]$ and $z \in \Uc \cap \C_+$ we have 
\[
\nr{\pppg x^{M} \Sp^{(n)}(z) \pppg x ^{-\d}}_{\Lc(L^2)} \leq C\left( 1 + \abs z ^{\frac d 2 - \e - 1 - n} \right),
\]
where $\ybar = (\y_0,\yyy)$.
\end{proposition}

\begin{proof}
The Proposition is a consequence of Proposition \ref{prop-low-freq-petite-perturbation}, the boundedness of $\p(\Pg)$ in $L^{2,\d}$ and the boundedness of $\pppg x ^M \Kz (1-\D) \inv \pppg x^\d$.
\end{proof}

\begin{remark} \label{rem-yy-yyy}
Until now we had not used the distinction between $\yy$ and $\yyy$. However the size of $\pppg x^M \Kzyy$ depends on $\yy$, so $\yy$ has to be fixed in order to obtain uniform estimates in Proposition \ref{prop-estim-S} and in Proposition \ref{prop-S-eps} below. On the other hand we have to keep the possibility to take $\yyy$ small. More precisely the choice of the cut-off function $\p$ in Proposition \ref{prop-S-eps} (and hence in the proof of Proposition \ref{prop-low-freq-perturbation-inter}) will depend on $\yy$, and then the choice of $\yyy$ will in turn depend on $\p$. This is why we could not simply take $\yy = \yyy$ in the definition of $\Hii$.
\end{remark}

\begin{proposition} \label{prop-S-eps}
Let $\y_0 \in ]0,1]$ be given by Proposition \ref{prop-low-freq-petite-perturbation}. Let $\e_1 > 0$, $\s > 2$ and $M \geq 0$. Then there exist a bounded neighborhood $\Uc$ of 0 in $\C$, $\p \in C_0^\infty(\R)$ equal to 1 on a neighborhood of 0 and $\yoo \in ]0,\y_0]$ such that for $\yyy \in ]0,\yoo]$ and $z \in \Uc \cap \C_+$ we have 
\[
\nr{\pppg x ^M \Spz \pppg x^{-\s}}_{\Lc(L^2)} \leq \e_1,
\]
where $\ybar = (\y_0,\yyy)$.
\end{proposition}

\begin{proof}
According to the Hardy inequality we have for $u \in \Sc$
\begin{align*}
\nr{\pppg x^M \Pcco u}_{L^2}
& \lesssim \sum_{j,k=1}^d \nr{\pppg x^M \big(D_j (\h_{\y_0} G_{j,k}) \big) D_k u}_{L^2} + \sum_{j,k=1}^d\nr{\pppg x^M \h_{\y_0}  \big(G_{j,k} - \d_{j,k}\big)  D_j D_k u}_{L^2}\\
& \lesssim \nr{u}_{\dot H^2}
\end{align*}
and 
\[
\nr{\pppg x^M W u}_{L^2} \lesssim \sum_{j=1}^d \nr{\pppg x^M b_j D_j  u}_{L^2}  \lesssim \nr{u}_{\dot H^2}.
\]
According to the third statement of Proposition \ref{prop-Pii-Hii-D} we obtain for $\m > 0$
\begin{eqnarray*}
\lefteqn{\nr{\pppg x ^M \Kz \Ri (i\m) \p (\Pg) \pppg x^{-\s} u}_{L^2} \lesssim \nr{ \Ri (i\m) \p (\Pg) \pppg x^{-\s} u}_{\dot H^2}}\\
&& \lesssim \nr{ \Piio \Ri (i\m) \p (\Pg) \pppg x^{-\s} u}_{L^2}\\
&& \lesssim \nr{\p (\Pg) \pppg x^{-\s} u}_{L^2} + \m \nr{\Ri(i\m)\p (\Pg) \pppg x^{-\s} u}_{L^2}+ \nr{\Bai \Ri(i\m) \p(\Pg) \pppg x^{-\s} u} \\
&& \lesssim \nr{\p (\Pg) \pppg x^{-\s} u}_{L^2} + O(\yyy^{1+\rho})\nr{\pppg x^{-1-\rho} (-\D+1) \Ri(i\m)\p (\Pg) \pppg x^{-\s} u}_{L^2} .
\end{eqnarray*}
The term with the factor $\m$ is estimated by the analog of \eqref{estim-res-trivial} for $\Hii$. For the term involving $\Bai$ we have used the fact that 
\[
\nr{\Bai (-\D+1)\inv \pppg x^{1+\rho}}_{\Lc(L^2)} = \bigo {\yyy} 0 (\yyy^{1+\rho}).
\]
Let $\p_1 \in C_0^\infty(\R)$ be equal to 1 on [-1,1]. For $\p \in C_0^\infty(\R)$ supported in $]-1,1[$ we have $\p(P) \pppg x ^{-\s} = \p(P) \p_1 (P) \pppg x ^{-\s}$. The operator $\p_1(P) \pppg x ^{-\s}$ is compact. On the other hand, since 0 is not an eigenvalue of $P$ the operator $\p(P)$ goes weakly to 0 when the support of $\p$ shrinks to $\singl 0$. Thus we can find $\p$ equal to 1 on a neighborhood of 0 such that for $\m > 0$ and $\yyy$ small enough we have
\[
\nr{\pppg x ^M S_{\p,\ybar}(i\m) \pppg x^{-\s}}_{\Lc(L^2)} \leq \frac {\e_1} 2.
\]
Now let $\t \in \R$ and $\m > 0$. We have 
\begin{eqnarray*}
\lefteqn{\nr{\pppg x^M \Kz \big( \Ri(\t+i\m) - \Ri(i\m) \big) \p(\Pg) \pppg x^{-\s}}}\\
&& \leq \nr{\pppg x^M \Kz (-\D+1)\inv \pppg x^\s} \times  \int_0^\t \nr{\pppg x^{-\s} (-\D+1)  \Ri^2(\th + i\m) \pppg x^{-\s}}\, d\th.
\end{eqnarray*}
The first factor is bounded by pseudo-differential calculus, and the second factor is of size $O(\abs \t)$ according to Proposition \ref{prop-low-freq-petite-perturbation}. Thus this norm is not greater that $\frac {\e_1} 2$ if $\t$ is small enough, and the proposition is proved.
\end{proof}

For $z \in \C_+$ and $\yyy \in ]0,1]$ we set 
\begin{equation} \label{def-Ro}
\Ro(z) = \big( \Pg  - z\big) \inv
\end{equation}
and
\[
\tRi(z) = \big( \Pg - i\Bai - z\big) \inv.
\]
In the following proposition we prove the resolvent estimates for $\tRi(z)$. Then we will add the contribution of $\Bac$ in the dissipative part to conclude the proof of Theorem \ref{th-low-freq}.

\begin{proposition} \label{prop-low-freq-perturbation-inter}
Let $\e > 0$ and $n \in \N$. Let $\d$ be as in the statement of Theorem \ref{th-low-freq}. Then there exist $\yyy$, $C\geq 0$ and a neighborhood $\Uc$ of $0$ in $\C$ such that for $z \in \Uc \cap \C_+$ and $\b_l,\b_r \in \R_+$ with $\b_l + \b_r \leq 2$ we have 
\[
\nr{\pppg x^{-\d} \pppg D^{\b_l} \tRi^{n+1}(z) \pppg D^{\b_r} \pppg x ^{-\d}}_{\Lc(L^2)} \leq C \left( 1 + \abs z ^{\frac d 2 - \e - 1 - n} \right).
\]
\end{proposition}

\begin{proof}
As for Proposition \ref{prop-low-freq-petite-perturbation} we see that it is enough to consider the case $\b_l = \b_r = 0$. Let $\s = \max(\d,3)$. Let $\e_1 \in\big] 0,\frac 14\big]$ and consider $\p \in C_0^\infty(\R)$ as given by Proposition \ref{prop-S-eps} for $M = \s$. We set $\BB(z) = \Ro(z) (1-\p)(\Pg)$. For any $\g \in \R$, this operator and its derivatives are uniformly bounded on $L^{2,\g}$ for $z \in \C_+$ close to 0. Let $\y_0$ be given by proposition \ref{prop-low-freq-petite-perturbation}. For $\yyy \in ]0,\y_0]$ we write $\ybar$ for $(\y_0,\yyy)$. We have 
\[
\tRi(z) = \Ri(z) \p(\Pg) - \tRi(z) \Sp(z) + \BB(z) + i \tRi(z) \Bai \BB(z),
\]
and hence for $n \in \N$
\begin{equation} \label{decomp-tRi}
\begin{aligned}
\tRi^{(n)}(z)
& = \Ri^{(n)}(z) \p(\Pg) + \tRi^{(n)}(z) \big( - \Sp(z) + i \Bai \BB(z) \big) + \BB^{(n)}(z)\\
& \quad  + i \sum_{j=0}^{n-1} C^n_j \tRi^{(j)}(z)\big( - \Sp^{(n-j)}(z) + i \Bai \BB^{(n-j)}(z) \big) .
\end{aligned}
\end{equation}
We prove by induction on $n \in \N$ that 
\begin{equation} \label{estim-tRi}
\nr{\pppg x^{-\d} \tRi^{(n)}(z) \pppg x^{-\s}} \lesssim 1 + \abs{z}^{\frac d 2 - \e - n - 1}.
\end{equation}
According to Propositions \ref{prop-low-freq-petite-perturbation}, \ref{prop-estim-S} and \ref{prop-S-eps}, the fact that $\p(\Pg)$ is uniformly bounded on $L^{2,\s}$ and the inductive induction for the sum in \eqref{decomp-tRi} (it vanishes if $n=0$), there exists $C \geq 0$ such that for $z \in \C_+$ close to 0 we have 
\[
\nr{\pppg x^{-\d} \tRi^{(n)}(z) \pppg x^{-\s}} \left(1 - \e_1 - \nr{\pppg x^{\s} \Bai \BB(z) \pppg x^{-\s} } \right) \leq C \left( 1 + \abs{z}^{\frac d 2 - \e - n - 1} \right).
\]
By pseudo-differential calculus we see that the norm of $\pppg x^{\s} \Bai \BB(z) \pppg x^{-\s}$ goes to 0 when $\yyy$ goes to 0. Thus if $\yyy$ is small enough we have 
\[
1 - \e_1  - \nr{\pppg x^{\s} \Bai \BB(z) \pppg x^{-\s} } \geq \frac 12,
\]
which concludes the proof of \eqref{estim-tRi}. In order to replace $\s$ by $\d$ we use \eqref{decomp-tRi} again and, estimating the second term with \eqref{estim-tRi} and Proposition \ref{prop-estim-S} instead of Proposition \ref{prop-S-eps} we obtain
\[
\nr{\pppg x^{-\d} \tRi^{(n)}(z) \pppg x^{-\d}} \left(1  - \nr{\pppg x^{\d} \Bai \BB(z) \pppg x^{-\d} } \right) \leq C \left( 1 + \abs{z}^{\frac d 2 - \e - n - 1} \right),
\]
and we conclude similarly.
\end{proof}

In order to prove Theorem \ref{th-low-freq} it remains to add the dissipative part with compactly supported absorption index. We begin with a lemma:

\begin{lemma} \label{lem-Rc0-Rc1}
Let $\Hc$ be a Hilbert space. Let $\Rc_0,\Rc_1 \in \Lc(\Hc)$ and let $\Bc$ be such that
\[
\Rc_1 = \Rc_0 - \Rc_0 \Bc \Rc_1 = \Rc_0 - \Rc_1 \Bc \Rc_0.
\]
Then for all $m \in \N$ we can write $\Rc_1^{m+1}$ as a linear combination of terms of the form 
\begin{equation}
\Rc_{0}^{m_1+1} \Bc \Rc_{j_2}^{m_2+1} \Bc \dots \Bc \Rc_{j_{k-1}}^{m_k+1} \Bc \Rc_0^{m_k + 1}
\end{equation}
where $k \in \N^*$, $j_1,\dots,j_{k-1} \in \{ 0,1\}$ and $m_1,\dots,m_k \in \N$ are such that 
\[ 
\sum_{l=1}^k m_l \leq m \quad \text{and} \quad m_l = 0 \quad \text{if} \quad j_l = 1.
\]
\end{lemma}

\begin{proof}
Using both of the identities between $\Rc_1$ and $\Rc_0$ we obtain 
\[
\Rc_1(z)  = \Rc_0(z) - \Rc_0(z)  \Bc \Rc_0(z)  + \Rc_0(z)   \Bc  \Rc_1(z)  \Bc  \Rc_0(z) .
\]
Then the result is proved by induction on $m$.
\end{proof}

Now we can finish the proof of Theorem \ref{th-low-freq}:

\begin{proof}[Proof of Theorem \ref{th-low-freq}]
Let $\yyy$ be given by Proposition \ref{prop-low-freq-perturbation-inter}.
Let $T = \pppg D^{\frac \a 2} a \in \Lc(H^1,L^2)$. We have $T^*T = \Ba$ so according to Corollary \ref{cor-quad-estim} we have 
\[
\nr{T R(z) T^*}_{\Lc(L^2)} \leq 1.
\]
Let $M \geq 0$ and $T_M = \pppg x^{-M} \pppg D^{\frac \a 2}$. We can write $\Bac = T_M^* \Bc_1 T = T^* \Bc_2 T_M = T_M^* \Bc_3 T_M$ where $\Bc_1$, $\Bc_2$ and $\Bc_3$ are bounded on $L^2$.
According to Lemma \ref{lem-Rc0-Rc1} applied with
\[
\Rc_0 = \tRi(z), \quad \Rc_1  = \Rg(z) \quad \text{and} \quad \Bc = \Bac,
\]
we can write $R^{m+1}(z)$ as a sum of terms of the form 
\begin{equation} \label{term-Rc1}
\Tc = \Rc_{0}^{m_1+1}(z) \Bc \Rc_{j_2}^{m_2+1}(z) \Bc \dots \Bc \Rc_{j_{k-1}}^{m_k+1}(z) \Bc \Rc_0^{m_k + 1}(z),
\end{equation}
where $k \in \N^*$, $j_1,\dots,j_{k-1} \in \{ 0,1\}$ and $m_1,\dots,m_k \in \N$ are such that $\sum_{l=1}^k m_l \leq m$ and $m_l = 0$ if $j_l = 1$.
If $M$ is large enough we obtain for such a term 
\begin{align*}
\nr {\pppg x^{-\d} \Tc \pppg x^{-\d}}
& \lesssim \nr{\pppg x^{-\d} \tRi^{m_1+1}(z) T_M^*} \times  \prod _{\substack {l=2\\ j_l = 0}}^{k-1} \nr{T_M \tRi^{m_l+1}(z) T_M^*}\\
& \qquad  \times\prod _{\substack{l=2 \\ j_l = 1}}^{k-1} \nr{T R(z) T^*} \times  \nr{T_M \tRi^{m_k+1} (z) \pppg x^{-\d} } \\
& \lesssim \prod_{l=1}^k \left( 1 + \abs z^{\frac d 2 - 1 - m_l - \e}\right)\\
& \lesssim \left( 1 + \abs z^{\frac d 2 - 1 - m - \e}\right).
\end{align*}
This concludes the proof of Theorem \ref{th-low-freq}.
\end{proof}

\subsection{Sharp low frequency resolvent estimate} \label{sec-low-freq-sharp}

We finish this section with the proof of Theorem \ref{th-low-freq-sharp}. The result follows from the self-adjoint analog by a simple perturbation argument, using the quadratic estimates and the spatial decay of the dissipative term:

\begin{proof}[Proof of Theorem \ref{th-low-freq-sharp}]
According to the resolvent identity, Proposition \ref{prop-quad-estim} and Theorem 1.1 in \cite{boucletr} we have
\begin{align*}
\nr{\pppg x \inv \Rg(z) \pppg x \inv}
& = \nr{\pppg x \inv \Ro(z) \pppg x \inv} + \nr{\pppg x\inv \Ro(z) \sqrt{\Ba}} \nr{\sqrt {\Ba} \Rg(z) \pppg x \inv }\\
& \lesssim 1 + \nr{\pppg x\inv \Ro(z) \sqrt{\Ba}} \nr{\pppg x \inv \Rg(z) \pppg x \inv}^\frac 12.
\end{align*}
Moreover,
\begin{align*}
\nr{\pppg x\inv \Ro(z) \sqrt{\Ba}}
& \leq \nr{\pppg x\inv \Ro(i) \sqrt{\Ba}} + \abs{z-i} \nr{\pppg x\inv \Ro(z) \pppg x\inv} \nr{\pppg x \Ro(i) \sqrt{\Ba}}\\
& \lesssim 1.
\end{align*}
For the norms involving $\Ro(i)$ we have used the fact that $\pppg x^\s \Ro(i) \sqrt {Ba}$ extends to a bounded operator since for $\s \leq 1$ and $u \in \Sc$ we have by pseudo-differential calculus
\[
\nr{\sqrt{\Ba} \Ro(i) \pppg x^\s u}_{L^2}^2 \leq \innp{\pppg x^\s \Ro(-i) \Ba \Ro(i) \pppg x^\s u}u \lesssim \nr{u}_{L^2}^2.
\]
This gives
\[
\nr{\pppg x \inv \Rg(z) \pppg x \inv} \lesssim 1 + \nr{\pppg x \inv \Rg(z) \pppg x \inv}^{\frac 12},
\]
from which the conclusion follows.
\end{proof}

\section{High frequency estimates} \label{sec-high-freq}

In this section we prove Theorem \ref{th-high-freq}. To this purpose we use semiclassical analysis (see for instance \cite{zworski}). For $h > 0$ and $\z \in \C_+$ we set $\Hh = h^2 \Hg$, $\Ph = h^2 \Pg$ and $R_h(\z) = (\Hh -\z)\inv$. Then for $n\in\N$, $z \in \C_+$ and $h = \abs z^{-\frac 12}$ we have 
\begin{equation} \label{eq-Rz-Rh}
\Rg(z)^{n+1} = \frac 1 {\abs z^{n+1}} R_h^{n+1}(\hat z) = h^{2(n+1)} R_h^{n+1}(\hat z)
\end{equation}
(we recall that $\hat z = z / \abs z$).\\

In order to prove uniform estimates for the resolvent $R_h(z)$ we use again the Mourre method. For high frequencies and in a dissipative context we follow \cite{art-mourre,boucletr14}. Here we have to be careful with the form of the dissipative part $h^2 \Ba$.

Let $\h_\a \in C_0^\infty(\R)$ be positive in a neighborhood of 1 and such that $0 \leq \h_\a(r) \leq r^{\frac \a 2}$ for all $r \in \R_+$. For $h \in ]0,1]$ we set 
\[
\tBa =  a(x) \h_\a \big(-h^2\D \big) a(x).
\]
Then we have
\begin{equation} \label{ineq-Ba}
0 \leq h^{2-\bb} \tBa  \leq h^{2-\a} \tBa \leq h^2 (-\D)^{\frac \a 2} \leq h^2 \Ba,
\end{equation}
in the sense that for all $\f \in H^{\a/2}(\R^d)$ we have 
\begin{equation} \label{comp-Ba-tBa}
0 \leq h^{2-\bb} \innp{\tBa \f}\f_{L^2(\R^d)} \leq h^2 \innp{\Ba \f}\f_{L^2(\R^d)}.
\end{equation}
The operator $\tBa$ is a bounded pseudo-differential operator on $L^2$. Its principal symbol is 
\[
b(x,\x) = a(x)^2 \h_\a(\abs \x^2). 
\]
The damping assumption \eqref{hyp-damping} on bounded trajectories is satisfied with $b$ instead of $a$:
\[
\forall w \in \O_b, \exists T \in\R, \quad b \big( \vf^T (w) \big) > 0. 
\]

Set 
\[
f_0(x,\x ) = x\cdot \x.
\]
As in \cite{boucletr14} (see Proposition 8.1), we can prove that there exist an open neighborhood $\tilde J$ of 1, $f_c \in C_0^\infty(\R^{2n},\R)$, $\b \geq 0$ and $c_0 > 0$ such that on $p\inv(\tilde J)$ we have 
\begin{equation} \label{eq-minor-symbol}
\{ p, f_0 + f_c \} + \b b  \geq 3 c_0,
\end{equation}
where $\{p,q\}$ is the Poisson bracket $\nabla_\x p \cdot \nabla_x q - \nabla_x p \cdot \nabla _\x q$. The fact that the symbol of the dissipative part depends on $\x$ does not change anything in the proof of this statement.
We set 
\[
F_h = \Opw (f_0 + f_c),
\]
where $\Opw$ is the Weyl quantization:
\[
\Opw(q) u(x) = \frac 1 {(2\pi h)^n} \int_{\R^n} \int_{\R^n} e^{\frac ih \innp{x-y}\x} q \left( \frac {x+y} 2, \x \right) u(y) \, dy \, d\x.
\]

Let $J$ be a neighborhood of 1 and a compact subset of $\tilde J$. Let $\h \in C_0^\infty(\tilde J,[0,1])$ be equal to 1 on a neighborhood of $J$.
After multiplication by $(\h \circ p)^2$, the (easy) G\aa rding inequality (Theorem 4.26 in \cite{zworski}) gives for $h>0$ small enough
\[
\Op \Big( (\h \circ p)^2 \{ p, f_0+ f_c \} + \b b (\h \circ p)^2 + 3c_0 \big(1-(\h\circ p)^2\big) \Big) \geq 3c_0 - O(h) \geq 2 c_0.
\]
After multiplication by $h^{2-\bb}$ we obtain 
\[
\h (\Ph) \big( [\Ph , ih^{1-\bb} F_h] + \b h^{2-\bb} \tBa \big) \h (\Ph) + 3c_0 (1-\h^2)(\Ph) \geq 2 c_0 h^{2-\bb} - O(h^{3-\bb}).
\]
After conjugation by $\1 J (\Ph)$ we obtain for $h$ small enough
\[
\1 J (\Ph) \big( [\Ph , ih^{1-\bb} F_h] + \b h^{2-\bb} \tBa \big) \1 J (\Ph) \geq c_0 h^{2-\bb} \1 J (\Ph).
\]
According to \eqref{ineq-Ba} this finally gives
\begin{equation} \label{minor-mourre}
\1 J (\Ph) \big( [\Ph , ih^{1-\bb} F_h] + \b h^2 \Ba \big) \1 J (\Ph) \geq c_0 h^{2-\bb} \1 J (\Ph), 
\end{equation}
which is the main assumption of Definition \ref{def-mourre} with $\b h^2$ instead of $\b$ and $\a =c_0 h^{2-\bb}$.
\\

It remains to check the other assumptions of Definition \ref{def-mourre}. The first is proved as in \cite{boucletr14} (except that we look at the norm in the form domain $H^1$ instead of the domain $H^2$), and the commutator properties are proved using (standard) pseudo-differential calculus, considering $h$ as a parameter (for the dissipative part we cannot use $h^{2-\a} \tBa$ as above, so we have to control directly the commutators of $h^2 \Ba$ with $h^{1 - \bb} F_h$) .\\

Thus we have proved that for $h \in ]0,h_0]$ the operator $h^{1-\bb} F_h$ is a conjugate operator to $\Hh$ on a neighborhood $J$ of 1 with lower bounds $h^{2-\bb} c_0$ for some $c_0 > 0$. According to Theorem \ref{th-mourre} we have proved the following result with $\pppg {F_h}^{-\d}$ instead of $\pppg x^{-\d}$:

\begin{proposition} \label{prop-high-freq-Rh}
Let $n \in \N$ and $\d > n + \frac 12$. There exists a neighborhood $J$ of $1$, $h_0 > 0$ and $C \geq 0$ such that for all $\z \in \C_+$ with $\Re(\z) \in J$ we have 
\[
\nr{\pppg x^{-\d} R_h^{n+1}(\z) \pppg x^{-\d}}_{\Lc(L^2)} \leq \frac C {h^{(2-\bb)(n+1)}}.
\]
\end{proposition}

In order to have the estimate with $\pppg x^{-\d}$ we proceed as usual (see the end of Section \ref{sec-inter-freq} for intermediate frequencies or \cite{art-mourre} in the semi-classical context). With \eqref{eq-Rz-Rh} and Proposition \ref{prop-high-freq-Rh} we obtain the second statement of Theorem \ref{th-high-freq}. For the first statement, we observe that under the non-trapping condition we can proceed as above with $\b = 0$ and with $\bb$ replaced by 1 in \eqref{minor-mourre}.

\section{Local energy decay} \label{sec-loc-decay}

In this section we use Theorems \ref{th-inter-freq}, \ref{th-low-freq} and \ref{th-high-freq} to prove Theorem \ref{th-loc-decay}. Let $u_0 \in \Sc$. We denote by $u$ the solution of \eqref{eq-damp-schrodinger}. Let $\m > 0$. For $t \in \R$ we set 
\[
u_\m(t) = \1 {\R_+}(t) u(t) e^{-t\m}.
\]
Then for $\t \in \R$ we set
\begin{equation} \label{def-check-u}
\check u_\m (\t) = \int_\R e^{it\t} u_\m (t) \, dt  =  \int_0^{+\infty} e^{it(\t+i\m)} u (t) \, dt ,
\end{equation}
so that for all $n \in \N$ and $\t \in \R$ we have 
\begin{equation} \label{check-u-n}
\check u_\m^{(n)} (\t) =  \int_\R (it)^n e^{it\t} u_\m (t) \, dt.
\end{equation}

We multiply \eqref{eq-damp-schrodinger} by $e^{it(\t+i\m)}$ and integrate over $\R_+$. This yields
\[
\big( \Hg - (\t+i\m) \big) \check u_\m (\t) = - i u_0
\]
and hence, for all $n \in \N$:
\begin{equation} \label{eq-der-umu}
\check u_\m ^{(n)} (\t) = - i \, n! \, \Rg^{n+1} (\t+i\m) u_0.
\end{equation}

\begin{lemma}
For all $n \in \N^*$ and $\m > 0$ the map 
$
\t \mapsto  \Rg^{n+1}(\t+i\m) u_0
$
belongs to $L^1(\R,L^2(\R^d))$.
\end{lemma}

\begin{proof}
Let $\h_0 \in C_0^\infty(\R,[0,1])$ be equal to 1 on a neighborhood of 0. According to \eqref{estim-Rg-trivial} the map $\t \mapsto  \Rg^{n+1} (\t +i\m) u_0$ is bounded, so it is enough to prove that $\t \mapsto (1-\h_0)(\t) \Rg^{n+1} (\t +i\m) u_0$ belongs to $L^1(\R)$. Let $z \in \C_+$. Using twice the identity 
\[
\Rg(z) = \frac {\Rg(z)(\Hg+1) - 1}{z+1}
\]
we get 
\[
\Rg(z)u_0 = \frac {1} {(z+1)^2} \Rg (z) (\Hg+1)^2 u_0 - \frac {1}{(z+1)^2}{(\Hg+1)u_0} - \frac 1 {z+1} u_0.
\]
The result follows after at least one differentiation with respect to $z$.
\end{proof}

This lemma does not provide any uniform estimate, but now we can take the Fourier transform of \eqref{check-u-n}. With \eqref{eq-der-umu} this gives for all $t \geq 0$:
\begin{equation} \label{eq-ut-inversee}
(it)^n e^{-t\m} u(t)  = - \frac {i  n!}{2\pi} \int_{\t\in\R} e^{-it\t}  \Rg^{n+1}(\t+i\m) u_0 \, d\t.
\end{equation}

We consider $\h_-,\h_0, \h \in C^\infty(\R, [0,1])$ such that $\h_-$ is supported in $]-\infty,0[$, $\h_0$ is compactly supported and equal to 1 on a neighborhood of 0, $\h$ is compactly supported in $]0,+\infty[$ and 
\[
\h_- + \h_0 + \sum_{j\in\N^*} \h_j = 1 \quad \text{ on $\R$,}
\]
where for $j \in \N^*$ and $\t \in \R$ we have set $\h_j (\t) = \h(\t/ 2^{j-1})$. We set $\h_+ = \sum_{j\in\N^*} \h_j$.
Starting from \eqref{eq-ut-inversee} applied with $n  = \kk -1$ ($\kk$ was defined in \eqref{def-talpha-kappa}) we can write
\begin{equation} \label{dec-u-v}
 u_\m(t) = -\frac {i n!} {2\pi (it)^{\k-1}} \big(v_-(t) + v_0(t) + v_+(t)\big)
\end{equation}
where for $* \in \{ - ,0,+\}$ we have set 
\begin{equation} \label{def-v-star}
v_*(t) =  \int_{\t\in\R} \h_*(\t) e^{-it\t}  \Rg^\k(\t+i\m) u_0 \, d\t.
\end{equation}

To simplify the notation we forget the dependance in $\m$. From now on, all the quantities depend on $\m > 0$ but the estimates are uniform in $\m$.

\begin{proposition} \label{prop-v-}
Let $k \in \N$. There exists $C \geq 0$ which does not depend on $u_0 \in \Sc$ such that for all $\m > 0$ and $t \geq 0$ we have 
\[
\nr{v_-(t)}_{L^2} \leq C \pppg t^{-k} \nr{u_0}_{L^2}.
\]
\end{proposition}

This implies that the corresponding contribution for $u(t)$ decays like any power of $t$ in $L^2$.

\begin{proof}
After $k$ partial integrations in \eqref{def-v-star} we get 
\[
(it)^k v_-(t) =  \int_\R e^{-it\t}  \frac {d^k}{d\t^k} \big( \h_-(\t) R(\t+i\m)^\k  \big) u_0 \, d\t.
\]
According to Remark \ref{rem-dissipative-accretive} we have
\[
\nr{ \frac {d^k}{d\t^k} \big( \h_-(\t) R(\t+i\m)^\k  \big)}_{\Lc(L^2)} \lesssim \pppg \t^{-(\k + k)},
\]
and the result follows.
\end{proof}

We now deal with $v_0$. The following result is (a slightly modified version of) Lemma 4.3 in \cite{boucletr14}:

\begin{lemma} \label{lem-holder2}
Let $\Hc$ be a Hilbert space. Let $f \in C^1 (\R^*, \Hc)$ be equal to 0 outside a compact subset of $\R$. Assume that for some $\g \in ]0,1[$ and $M_f \geq 0$ we have
\[
\forall \t \in \R^*, \quad  \nr{f(\t)}_\Hc \leq M_f \abs \t^{-\g} \quad \text{and} \quad \nr{f'(\t)}_\Hc \leq M_f \abs \t ^{-1 -\g}.
\]
Let $\b \in [0,1[$. Then there exists $C \geq 0$ which does not depend on $f$ and such that for all $t \in \R$ we have
\[
\nr{\hat f (t)}_\Hc \leq C \, M_f \, \pppg t ^{\b(\g-1)}.
\]
\end{lemma}

\begin{proof}
Following the proof of \cite{boucletr14} we set $f_t(\t) = \int_{-1}^1 f \big(\t - \frac s t\big) \, ds$ where $\vf \in C_0^\infty(]-1,1[,\R)$ satisfies $\int_\R \vf = 1$ and we write for $\abs t \geq 1$
\begin{align*}
\abs{\hat f(t)}
& \leq \int_{\abs \t \leq t^{-\b}} \nr{f(\t)} \, d\t + \int_{\abs \t \geq t^{-\b}} \nr{f(\t) - f_t(\t)} \, d\t + \nr{\int_{\abs \t \geq t^{-\b}} e^{-it\t} f_t(\t) \, d\t}\\
& \lesssim \abs t^{\b(1-\g)} + \abs {t}^{\g \b -1} + \frac 1 t \left( \nr{f_t(t^{-\b})} + \nr{f_t(-t^{-\b})} + \nr{\int_{\abs \t \geq t^{-\b}} e^{-it\t} f_t'(\t) \, d\t} \right)\\
& \lesssim \abs t^{\b(1-\g)}.
\end{align*}
We omit the details.
\end{proof}

\begin{proposition}
Let $\e \in \big]0,\frac 12\big[$ and $\d > \k + \frac 12$. Then there exists $C \geq 0$ which does not depend on $u_0 \in \Sc$ and such that for all $\m > 0$ and $t \geq 0$ we have 
\[
\nr{v_0(t)}_{L^{2,-\d}} \leq \pppg t^{\k-1 - \frac d 2  + \e} \nr{u_0}_{L^{2,\d}}.
\]
\end{proposition}

\begin{proof}
According to Theorem \ref{th-low-freq} applied with $\e /2$ instead of $\e$ and Theorem \ref{th-inter-freq} there exists $C \geq 0$ (which does not depend on $u_0$) such that for $\m > 0$, $\t \in \R$ and $z = \t + i\m$ we have
\[
\nr{\h_0(\t) \Rg^\k(z) u_0}_{L^{2,-\d}} \leq C \abs{z}^{\frac d 2 - \k - \frac \e 2} \nr{u_0}_{L^{2,\d}}
\]
and
\[
\nr{\frac d {d\t} \big( \h_0(\t) \Rg^\k(z) \big) u_0}_{L^{2,-\d}} \leq C \abs{z}^{\frac d 2 - \k -1 - \frac \e 2} \nr{u_0}_{L^{2,\d}}.
\]
Then the statement follows from Lemma \ref{lem-holder2} applied with $\b \in ]0,1[$ so close to 1 that 
\[
\b \big(\k  - \frac d 2 - 1+ \frac \e 2\big) \leq \k - \frac d 2 - 1 +  \e. \qedhere
\]
\end{proof}

To finish the proof of Theorem \ref{th-loc-decay} we have to estimate $v_+(t)$. As for $v_-(t)$ above, $k$ partial integrations yield
\begin{align*}
(it)^k v_+(t)
& = \int_\R e^{-it\t}  \sum_{j=1}^{k} C_k^j  \h_+^{(j)} (\t) \Rg^{\k+k-j}(\t+i\m)  u_0 \, d\t\\
& \quad + \int_\R  e^{-it\t}  \h_+ (\t) \Rg^{\k+k}(\t+i\m) u_0 \, d\t\\
& =: v_{+,k}^0(t) + w_k(t)
\end{align*}

The following proposition proves that the contribution of $v_+(t)$ in \eqref{dec-u-v} decays like any power of $t$. However there may be a loss of two derivatives when $\a = 0$ if the non-trapping assumption does not hold.

\begin{proposition} \label{prop-v+}
Let $k \in \N^*$ and $\d > \k + k - \frac 12$. Let $\s \in [0,2]$.
\begin{enumerate} [(i)]
\item \label{item-v+i} 
There exists $C \geq 0$ which does not depend on $u_0$ and such that for all $\m > 0$ and $t \geq 1$ we have 
\[
\nr{\pppg x^{-\d} v_{+,k}^0(t)}_{L^2} \leq C \nr{u_0}_{L^{2,\d}}.
\]
\item \label{item-v+ii} 
Assume that the non-trapping assumption \eqref{hyp-non-trapping} holds or that we have the damping condition \eqref{hyp-damping} together with $(\kk+k) \bb + \s  \geq 2$. Then there exists $C \geq 0$ which does not depend on $u_0$ such that for all $\m > 0$ and $t \geq 1$ we have 
\[
\nr{\pppg x^{-\d} w_{k}(t)}_{L^2} \leq C \nr{u_0}_{H^{\s,\d}}.
\]
\end{enumerate}
\end{proposition}

\begin{proof}
\stepp Statement \eqref{item-v+i} follows from Theorem \ref{th-inter-freq} and the fact that $\h_+^{(j)}$ is compactly supported in $]0,+\infty[$ for all $j \geq 1$. 

\stepp Assume that $\a \geq 1$ (and hence $\bb = 1$) or that the non-trapping condition holds. Then according to Theorem \ref{th-high-freq} we have for $\t \in \supp(\h_+)$
\[
\nr{\pppg x^{-\d} \Rg^{\k+k}(\t+i\m) \pppg x^{-\d}}_{\Lc(L^2)} \lesssim \abs{\t}^{-\frac {\k+k}2}.
\]
Since $\k + k \geq 3$ this gives the second statement with $\s = 0$.

\stepp Now assume that $\a \in [0,1[$. For $j \in \N^*$ we set
\begin{equation*} 
w_{k,j}(t) =  \int_{\t\in\R} \h_j(\t) e^{-it\t}  \Rg^{\k+k}(\t+i\m) u_0 \, d\t.
\end{equation*}
Let $\tilde \h \in C_0^\infty(\R_+^*,[0,1])$ be equal to 1 on a neighborhood of $\supp \h$. For $\t \in \R$ and $j \in \N^*$ we set $\tilde \h_j(\t) = \tilde \h(\t/2^{j-1})$. Let 
\[
I_{k,j}(t) =  \int_{\t\in\R} \h_j(\t) e^{-it\t}  \pppg x^{-\d}  \Rg^{\k+k}(\t+i\m) \pppg x^{-\d}  \, d\t \quad \in \Lc(L^2).
\]
We have 
\[
\pppg x^{-\d} w_{k,j}(t) = w_{k,j}^1(t) + w_{k,j}^2(t) + w_{k,j}^3(t)
\]
where 
\[
w_{k,j}^1(t) = \tilde \h_j(\Pg) I_{k,j}(t)  \tilde \h_j(\Pg) \pppg x^{\d}  u_0,
\]
\[
w_{k,j}^2(t) = (1-\tilde \h_j)(\Pg) I_{k,j}(t)  \tilde \h_j(\Pg) \pppg x^\d  u_0
\]
and
\[
w_{k,j}^3(t) =  I_{k,j}(t)   (1-\tilde \h_j)(\Pg)   \pppg x^\d u_0.
\]

\stepp By almost orthogonality, Theorem \ref{th-high-freq} and almost orthogonality again we have 
\begin{eqnarray*}
\lefteqn{\nr{\sum_{j\in\N^*}  w_{k,j}^1(t)}^2 \lesssim \sum_{j\in\N^*} \nr{ w_{k,j}^1(t)}^2}\\
&& \lesssim \sup_{j \in \N^*} \left( \int_{\t\in\R} \h_j(\t) \nr{ \pppg x^{-\d}  \Rg^{\k+k}(\t+i\m) \pppg x^{-\d}} \, d\t \right)^2 \times \sum_{j\in\N^*} \nr{\tilde \h_j(\Pg) \pppg x^{\d}  u_0}^2\\
&& \lesssim \sup_{j \in \N^*}  2^{2j} 2^{- {j(\k+k)\a}} 2^{-  {j\s}} \nr{\pppg{\Pg}^{\frac {\s}2} \pppg x^\d u_0}^2\\
&& \lesssim \nr{u_0}_{H^{\s,\d}}^2.
\end{eqnarray*}
It remains to prove that 
\begin{equation} \label{estim-w2w3}
\nr{w_{k,j}^2(t)} + \nr{w_{k,j}^3(t)}  \lesssim 2^{-j} \nr{u_0}_{L^{2,\d}}.
\end{equation}

\stepp 
For the contribution of $w_{k,j}^2(t)$ we prove that there exists $C \geq 0$ such that for $j \in \N^*$ and $\t \in \supp(\h_j)$ we have 
\begin{equation} \label{estim-blip}
\nr{(1-\tilde \h_j)(\Pg) \pppg x^{-\d} \Rg^{\k + k}(\t + i\m) \pppg x^{-\d}}_{\Lc(L^2)} \leq C 2^{-2j}.
\end{equation}
Let $\z \in C_0^\infty(\R_+^*,[0,1])$ be equal to 1 on a neighborhood of $\supp \h$ and such that $\tilde \h = 1$ on a neighborhood of $\supp \z$. For $j \in \N^*$ and $\t \in \R$ we set $\z_j(\t) = \z(\t/2^{j-1})$. According to Theorem 8.7 in \cite{dimassis} about functions of a self-adjoint semiclassical pseudo-differential operator (with $h = 2^{-\frac{j-1}2}$) we see that for any $M \geq 0$ we have
\[
\nr{(1-\tilde \h _j)(\Pg) \pppg x^{-\d} \z_j(\Pg) \pppg x^\d}_{\Lc(L^2)} \lesssim 2^{-jM}.
\]
Since $(1-\tilde \h_j)(\Pg)$ is uniformly bounded, it is remains to prove 
\begin{equation*} 
\nr{ \pppg x^{-\d}(1- \z_j)(\Pg) \Rg^{\kk + k}(\t + i\m) \pppg x^{-\d}}_{\Lc(L^2)} \lesssim 2^{-2j}.
\end{equation*}
We recall that for $z \in \C_+$ we have set $\Ro(z) = (\Pg-z)\inv$. By the resolvent identity we have
\begin{equation*} 
\begin{aligned}
\Rg^{\kk+k}(z)
& = \Ro(z)  \Rg^{\kk + k-1}(z) + i \Ro(z) \Ba \Rg^{\kk+k}(z) \\
& = \Ro(z)^2 \Rg^{\kk + k -2}(z) + i \Ro^2(z) \Ba \Rg^{\kk + k -1}(z) + i \Ro(z) \Ba \Rg^{\kk+k}(z).
\end{aligned}
\end{equation*} 
For $\Re(z) \in \supp(\h_j)$ we have 
\begin{eqnarray*}
\lefteqn{\nr{ \pppg x^{-\d}(1- \z_j)(\Pg) \Ro^2(z)  \Rg^{{\kk+k}-2}(z) \pppg x^{-\d}}}\\
&&\leq \nr{ \pppg x^{-\d}(1- \z_j)(\Pg) \Ro^2(z)\pppg x^\d} \nr{  \pppg x^{-\d}  \Rg^{{\kk+k}-2}(z) \pppg x^{-\d}}\\
&& \lesssim 2^{-2j}.
\end{eqnarray*}
We have used the Spectral Theorem and pseudo-differential calculus to estimate the first factor, and Theorem \ref{th-high-freq} for the second. Similarly
\begin{eqnarray*}
\lefteqn{\nr{ \pppg x^{-\d}(1- \z_j)(\Pg) \Ro^2(z) \Ba  \Rg(z)^{\kk+k-1} \pppg x^{-\d}}}\\
&&\leq \nr{ \pppg x^{-\d}(1- \z_j)(\Pg) \Ro^2(z) \Ba \pppg x^\d} \nr{  \pppg x^{-\d}  \Rg(z)^{\kk+k-1} \pppg x^{-\d}}\\
&& \lesssim 2^{j \left( -2 + \frac \a 2 \right) - \frac {\a j} 2 (\kk + k-1)}\lesssim 2^{-2j}.
\end{eqnarray*}
For the last term we have to prove 
\begin{equation} \label{estim-zeta}
\nr{ \pppg x^{-\d}(1- \z_j)(\Pg)\Ro(z) \Ba \Rg^{\kk+k}(z) \pppg x^{-\d}} \lesssim 2^{-2j}.
\end{equation}
We proceed as above. We consider $\vf \in C_0^\infty(\R_+^*,[0,1])$ equal to 1 on a neighborhood of $\supp(\h)$ and such that $\z =1$ on a neighborhood of $\supp(\vf)$, and then we set $\vf_j = \vf (\cdot / 2^{j-1})$ for $j \in \N^*$. We write 
\begin{eqnarray*}
\lefteqn{\nr{\pppg x^{-\d}(1- \z_j)(\Pg)\Ro(z) \Ba \Rg^{\kk+k}(z) \pppg x^{-\d}}}\\
&& \leq  \nr{\pppg x^{-\d}(1- \z_j)(\Pg)\Ro(z) \Ba \vf_j(P) \pppg x^\d}\nr{\pppg x^{-\d} \Rg^{\kk+k}(z) \pppg x^{-\d}}\\
&& \quad + \nr{\pppg x^{-\d}(1- \z_j)(\Pg)\Ro(z) \Ba \pppg x^\d} \nr{\pppg x^{-\d} (1-\vf_j)(P) \Rg^{\kk+k}(z) \pppg x^{-\d}}\\
&& \lesssim 2^{-jM} + 2^{j \left(\frac \a 2 - 1 \right)} \nr{\pppg x^{-\d} (1-\vf_j)(P) \Rg^{\kk+k}(z) \pppg x^{-\d}},
\end{eqnarray*}
for any $M \geq 0$. Then we use again the resolvent identity
\begin{align*}
\pppg x^{-\d} (1-\vf_j)(P) \Rg^{\kk+k}(z) \pppg x^{-\d}
& = \pppg x^{-\d} (1-\vf_j)(P) \Ro(z) \Rg^{\kk+k-1}(z) \pppg x^{-\d}\\
&\quad  + \pppg x^{-\d} (1-\vf_j)(P) \Ro(z) \Ba \Rg^{\kk+k}(z) \pppg x^{-\d}
\end{align*}
and we conclude as above. We obtain \eqref{estim-zeta}, then \eqref{estim-blip}, and finally the contribution of $w_{k,j}^2$ in \eqref{estim-w2w3} after integration over $\supp(\h_j)$. The contribution of $w_{k,j}^3$ is estimated similarly. Then it remains to sum over $j \in \N^*$ to conclude the proof of the proposition and hence the proof of Theorem \ref{th-loc-decay}. 
\end{proof}

\section{Smoothing effect} \label{sec-smoothing-effect}

In this section we prove Theorem \ref{th-smoothing-bis}. With Theorems \ref{th-inter-freq}, \ref{th-low-freq-sharp} and \ref{th-high-freq} it implies Theorem \ref{th-smoothing-effect}. For this we use a dissipative version of the theory of relatively smooth operators in the sense of Kato.

\begin{proposition} \label{prop-estim-regularisant}
Under the assumption of Theorem \ref{th-smoothing-bis} there exists $C \geq 0$ such that for all $z \in \C_+$ we have 
\[
\nr{ \pppg x \inv \pppg \Pg ^{\frac \bbb 4} \Rg(z) \pppg \Pg ^{\frac \bbb 4} \pppg x \inv }_{\Lc(L^2)} \leq C.
\]
\end{proposition}

\begin{proof}
\stepp Let $K$ be a compact subset of $\C$. Using the resolvent identity
\begin{align*}
\Rg(z) = \Rg(i) + (z-i) \Rg(i)^2 + (z-i)^2 \Rg(i) \Rg(z) \Rg(i),
\end{align*}
we obtain for $z \in \C_+ \cap K$
\begin{eqnarray*}
\lefteqn{\nr{ \pppg x \inv \pppg \Pg ^{\frac \bbb 4} \Rg(z) \pppg \Pg ^{\frac \bbb 4} \pppg x \inv}}\\
&& \lesssim 1 + \nr{ \pppg x \inv \pppg \Pg ^{\frac \bbb 4}\Rg(i) \pppg x} \nr{\pppg x \inv \Rg(z) \pppg x \inv} \nr{\pppg x \Rg(i) \pppg \Pg ^{\frac \bbb 4} \pppg x \inv}.
\end{eqnarray*}
By pseudo-differential calculus the operators $\pppg \Pg ^{\frac \bbb 4}\Rg(i)$ and $\Rg(i) \pppg \Pg ^{\frac \bbb 4}$ are bounded on $L^{2,-1}$ and $L^{2,1}$, respectively. For the second factor in the right-hand side we use \eqref{hyp-res-estim}, and the conclusion follows for $z \in \C_+ \cap K$.

\stepp It remains to prove the result for $\abs z \gg 1$. Let $\h \in C_0^\infty(\R,[0,1])$ be supported on [-3,3] and equal to 1 on [-2,2]. For $z \in \C_+$ we define $\h_z : \l \mapsto \h(\l/\abs z)$. The operator $\e^{\frac \bbb 4} \pppg {\Pg}^{\frac \bbb 4}  \pppg {\e \Pg}^{-\frac \bbb 4}$ is a pseudo-differential operator whose symbol has bounded derivatives uniformly in $\e \in ]0,1]$, so the operator 
\begin{equation} \label{op-1}
\abs {z}^{-\frac \bbb 4} \pppg x^{-1} \pppg {\Pg}^{\frac \bbb 4}  \pppg {\frac {\Pg}{\abs z}}^{-\frac \bbb 4} \pppg x
\end{equation}
extends to a bounded operator on $L^2$ uniformly in $z$ with $\abs z \geq 1$. The operator 
\begin{equation} \label{op-2}
\pppg x^{-1}  \pppg {\frac {\Pg}{\abs z}}^{\frac \bbb 4} \h\left( \frac \Pg {\abs z} \right) \pppg x
\end{equation}
is also bounded on $L^2$ uniformly in $z$ with $\abs z \geq 1$, and we have similar estimates for the adjoint operators of \eqref{op-1} and \eqref{op-2}. Thus
\begin{eqnarray*}
\lefteqn{\nr{\pppg x^{-1} \pppg {\Pg}^{\frac \bbb 4}  \h_z(\Pg) \Rg(z)  \h_z(\Pg) \pppg {\Pg}^{\frac \bbb 4}\pppg x^{-1}}}\\
&& \lesssim \abs z^{\frac \bbb 2} \nr{\pppg x^{-1} \pppg {\abs z \inv {\Pg}}^{\frac \bbb 4}  \h_z(\Pg) \Rg(z)  \h_z(\Pg) \pppg {\abs z \inv {\Pg}}^{\frac \bbb 4}\pppg x^{-1}}\\
&& \lesssim \abs z^{\frac \bbb 2} \nr{\pppg x^{-1}  \Rg(z) \pppg x^{-1}}\\
&& \lesssim 1.
\end{eqnarray*}

\stepp 
With $\Ro(z) = (\Pg-z)\inv$ we have the resolvent identity 
\[
\Rg(z) = \Ro(z) + i \Rg(z) \Ba \Ro(z).
\]
We have
\begin{eqnarray*}
\lefteqn{\nr{\pppg x^{-1} \pppg {\Pg}^{\frac \bbb 4}  \h_z(\Pg) \Ro(z)  (1-\h_z)(\Pg) \pppg {\Pg}^{\frac \bbb 4}\pppg x^{-1}}}\\
&& \leq \nr{\pppg x^{-1} \pppg {\Pg}^{\frac \bbb 4}\h_z(\Pg)} \nr{\Ro(z) (1- \h_z)(\Pg) \pppg {\Pg}^{\frac \bbb 4}\pppg x^{-1}}\\
&& \lesssim \pppg z^{\frac \bbb 4} \pppg z^{\frac \bbb 4 - 1}  \lesssim 1.
\end{eqnarray*}
We have estimated the first factor as above and the second by the Spectral Theorem.
On the other hand, since the operator $\sqrt {\Ba} \pppg {\Pg}^{-\frac 12}$ is bounded we also have by Proposition \ref{prop-quad-estim}
\begin{eqnarray*}
\lefteqn{\nr{\pppg x^{-1} \pppg {\Pg}^{\frac \bbb 4}  \h_z(\Pg) \Rg(z) \Ba \Ro(z)  (1-\h_z)(\Pg) \pppg {\Pg}^{\frac \bbb 4}\pppg x^{-1}}}\\
&& \leq \nr{\pppg x^{-1} \pppg {\Pg}^{\frac \bbb 4}\h_z(\Pg)\pppg x} \nr{\pppg x^{-1} \Rg(z) \sqrt {\Ba}} \nr{ \pppg {\Pg}^{\frac 12} \Ro(z) (1- \h_z)(\Pg) \pppg {\Pg}^{\frac \bbb 4}}\\
&& \lesssim \pppg z^{\frac \bbb 4} \pppg z^{-\frac \bbb 4} \pppg z^{\frac 12 + \frac \bbb 4 - 1}  \lesssim 1.
\end{eqnarray*}
This proves that 
\[
\nr{\pppg x^{-1} \pppg {\Pg}^{\frac \bbb 4}  \h_z(\Pg) \Rg(z)  (1-\h_z)(\Pg) \pppg {\Pg}^{\frac \bbb 4}\pppg x^{-1}} \lesssim 1.
\]

\stepp 
The operator 
\[
\pppg x^{-1} \pppg {\Pg}^{\frac \bbb 4} (1-\h_z) (\Pg) \Rg(z)  \h_z (\Pg) \pppg {\Pg}^{\frac \bbb 4}\pppg x^{-1}
\]
is estimated similarly. Finally for 
\[
\pppg x^{-1} \pppg {\Pg}^{\frac \bbb 4} (1-\h_z) (\Pg) \Rg(z)  (1-\h_z) (\Pg) \pppg {\Pg}^{\frac \bbb 4}\pppg x^{-1}
\]
we only have to use twice the resolvent identity:
\begin{align*}
\Rg(z)
& = \Ro(z) + i \Ro(z) \Ba \Ro(z) - \Ro(z) \Ba \Rg(z) \Ba \Ro(z).
\end{align*}
Then we apply the same idea as above, using Corollary \ref{cor-quad-estim} to estimate $\sqrt {\Ba} R(z) \sqrt {\Ba}$. This concludes the proof.
\end{proof}

Taking the adjoint in the estimate of proposition \ref{prop-estim-regularisant} we obtain the same estimate with $\Rg(z)$ replaced by $\Rg(z)^* = (\Pg +i\Ba - \bar z)\inv$ (the same is true for the estimates of Theorems \ref{th-inter-freq}, \ref{th-low-freq-sharp} and \ref{th-high-freq}). In particular we obtain the following result:

\begin{corollary} \label{cor-estim-regularisant}
Then there exists $C \geq 0$ such that for all $z \in \C_+$ and $\f \in \Sc$ we have 
\[
\abs{ \innp{ \big( (H-z)\inv - (H^* - \bar z)\inv \big) \pppg \Pg ^{\frac \bbb 4} \pppg x^{-1} \f}{ \pppg \Pg^{\frac \bbb 4} \pppg x^{-1} \f}_{L^2}} \leq C \nr \f_{L^2}^2.
\]
\end{corollary}

It is known that such an estimate on the resolvent implies Theorem \ref{th-smoothing-bis}. This comes from the dissipative version of the theory of relatively smooth operators. The self-adjoint theory can be found in \cite[\S XIII.7]{rs4}. The dissipative version uses the theory of self-adjoint dilations for a dissipative operator described in \cite{nagyf}. All this has been combined in Proposition 6.2 in \cite{art-mourre-formes}, according to which Theorem \ref{th-smoothing-bis} follows.

\bibliographystyle{alpha}
\bibliography{bibliotex}

\vspace{1cm}

\begin{center}

\begin{minipage}{0.45 \linewidth}

\noindent  {\sc Moez KHENISSI}

\noindent  {\sc \'Ecole Sup\'erieure des Sciences et de Technologie de Hammam Sousse}

\noindent  {\sc Rue Lamine El Abbessi}

\noindent  {\sc 4011 Hammam Sousse}

\noindent  {\sc Tunisia}

 \begin{verbatim}
moez.khenissi@fsg.rnu.tn
\end{verbatim}
\end{minipage}
\hfill
\begin{minipage}{0.45 \linewidth}

\noindent  {\sc Julien ROYER}

\noindent  {\sc Institut de mathématiques de Toulouse}

\noindent  {\sc 118, route de Narbonne}

\noindent  {\sc 31062 Toulouse Cédex 9}

\noindent  {\sc France}

 \begin{verbatim}
julien.royer@math.univ-toulouse.fr
\end{verbatim}
\end{minipage}
\end{center}

\end{document}